\documentclass[draft,12pt,reqno,a4paper]{amsart}

\usepackage[margin=3cm,footskip=1cm]{geometry}
\usepackage{amssymb} 
\usepackage{amsmath} 
\usepackage{amsthm} 
\usepackage{mathtools}
\usepackage{mathrsfs}

\usepackage{enumerate} 
\usepackage[latin1]{inputenc}
\usepackage[shortlabels]{enumitem}
\usepackage{color}
\usepackage{graphicx} 
\usepackage{verbatim}

\DeclareMathOperator*{\slim}{s-lim}

\newcommand{\N}{{\mathbb{N}}} 
\newcommand{\R}{{\mathbb{R}}} 
\newcommand{\C}{{\mathbb{C}}} 

\newcommand{\vB}{{\mathcal B}}
\newcommand{\vD}{{\mathcal D}}
\newcommand{\vE}{{\mathcal E}} 
\newcommand{\vF}{{\mathcal F}}
\newcommand{\vH}{{\mathcal H}}
\newcommand{\vL}{{\mathcal L}}
\newcommand{\vO}{{\mathcal O}}

\newcommand{\scD}{{\mathscr D}}
\newcommand{\scF}{{\mathscr F}}

\theoremstyle{plain}
\newtheorem{thm}{Theorem}[section]
\newtheorem{proposition}[thm]{Proposition}
\newtheorem{lemma}[thm]{Lemma} 
\newtheorem{corollary}[thm]{Corollary}
\theoremstyle{definition}

\newtheorem{cond}[thm]{Condition}

\newtheorem{remark}[thm]{Remark}

\newtheorem*{remarks*}{Remarks}
\newtheorem*{remark*}{Remark}

\numberwithin{equation}{section}

\title{Stationary scattering theory for repulsive Hamiltonians}

\author{Kyohei Itakura}
\thanks{Research Organization of Science and Technology, 
Ritsumeikan University\ 1-1-1 Noji-higashi, Kusatsu, shiga, 525-8577, Japan.\ 
E-mail: kitakura@gst.ritsumei.ac.jp}

\date{}
\begin{document}

\begin{abstract}
In the present paper 
we discuss stationary scattering theory for repulsive Hamiltonians.
We show the existence and completeness of stationary wave operators 
and unitarity of the scattering matrix. 
Moreover we completely characterize asymptotic behaviors of generalized eigenfunctions 
with minimal growth in terms of the scattering matrix. 
In our argument the radiation condition bounds for limiting resolvents play major roles. 
In fact, it is used to construct the stationary wave operators. 
\end{abstract}

\allowdisplaybreaks
\maketitle

\noindent
{\it Keywords:}
repulsive Hamiltonians,
stationary scattering theory, 
scattering matrix.

\tableofcontents


\section{Introduction}

The purpose of this paper is to establish stationary scattering theory for 
repulsive Hamiltonians given by 
\begin{equation*}
H_\alpha = \tfrac12p^2-\tfrac12|x|^\alpha+q(x) \quad\text{on}\ L^2(\R^d),
\end{equation*}
where $0<\alpha<2$, $p=-i\partial$ and $q$ is a perturbation. 
We assume that $q$ is a real-valued function belonging to $L^2_{\rm loc}(\R^d)$ 
and decays at infinity at a rate dependent on the parameter $\alpha$.
We will give a more precise condition on $q$ in Section~\ref{sec:Setting}.

It is well-known that the spectrum of $H_\alpha$ with short-range perturbation $q$ 
is purely absolutely continuous.
Moreover by \cite{BCHM} it was proved that the wave operators 
\begin{equation*}
W^\pm=\slim_{t\to\pm\infty}e^{itH_\alpha}e^{-itH_{\alpha,0}}
\end{equation*}
exist and are complete. 
Here we denote the \emph{free} repulsive Hamiltonian as $H_{\alpha,0}$, that is
\begin{equation*}
H_{\alpha,0} := \tfrac12p^2-\tfrac12|x|^\alpha.
\end{equation*}
However stationary scattering theory is not yet established 
for the repulsive Hamiltonians 
even for short-range perturbations, as far as the author knows.
In the paper, we deal with several topics on stationary scattering theory. 
The first one is existence and completeness of the stationary wave operators $\vF^\pm$. 
To construct the stationary wave operators 
we use the radiation condition bounds stated as Corollary~\ref{cor:rc-real} below  
and employ the schemes of \cite{GY, Iso, IS2}. 
The second one is unitarity of the scattering matrix $S(\cdot)$. 
The last one is asymptotic behaviors of \emph{minimal} generalized eigenfunctions 
(see Theorem~\ref{thm:characterize-B^*-ef}). 
Here we say minimal in the sense that the growth order at infinity is minimal. 
We obtain a characterization of them by outgoing/incoming spherical waves.  
We note these topics are not dealt with in \cite{BCHM}. 
In this sense our results are new, and this is a novelty of the paper. 

In this paper we discuss only the case of $0<\alpha<2$, 
although the case of $\alpha=2$ is included in \cite{BCHM}. 
When $\alpha=2$ 
the classical particles scatter with exponential order, 
although when $0<\alpha<2$ they scatter with polynomial order. 
Then our \emph{escape function}, which plays an important role in the study 
of repulsive Hamiltonians (cf. \eqref{eq:escape-function}), 
is defined by using a logarithmic function, see \cite{Ita1, Ita2}. 
Thus when $\alpha=2$,  
we need a more stronger result than Corollary~\ref{cor:rc-real} 
to construct the stationary wave operators.

In Section~\ref{sec:Setting} we state our setting and results. 
We introduce our basic setting, for example, definition of escape function 
and an assumption on $q$, 
and spectral theory, which is a refinement of the results of our previous papers, 
and state our main results on stationary scattering theory.
We are going to give the proofs in later sections. 
In Sections~\ref{sec:Wave operator} we discuss properties of stationary wave operators, 
and in Section~\ref{sec:Wave matrix and scattering matrix} 
we investigate a characterization of minimal generalized eigenfunctions.


\section{Setting and results}\label{sec:Setting}

\subsection{Basic setting}\label{sec:Basic setting}

We choose a smooth cut-off function $\chi$ which satisfies 
\begin{equation}\label{eq:chidef}
\chi(s)=
\begin{cases}
1 \quad \text{for }\ s \le 1, \\
0 \quad \text{for }\ s \ge 2,
\end{cases}
\quad
\tfrac{\mathrm d}{\mathrm ds}\chi=\chi'\le 0.
\end{equation}
Throughout the paper, we fix the function $\chi$.
By using the function $\chi$ we introduce the function $r\in C^\infty(\R^d)$, 
which is a modification of  $|x|$ on a neighborhood of the origin, by
\begin{equation*}
r=r(x) = \chi(|x|) + \bigl( 1-\chi(|x|) \bigr)|x|.
\end{equation*}
Now our escape function $f\in C^\infty(\R^d)$ is given as follows: 
\begin{equation}\label{eq:escape-function}
f(x) = 
\dfrac{r^{1-\alpha/2}-1}{1-\alpha/2} + 1. 
\end{equation}
Such a choice of $f$ is based on the scattering order of classical particles 
subject to the repulsive electric field (see \cite{Ita1,Ita2}). 
We note that $f\ge1$ on $\R^d$.

Throughout the paper we assume the following condition.
\begin{cond}\label{cond:short-range}
The perturbation $q$ is a real-valued function and belongs to $C^1(\R^d)$.
Moreover there exist $\rho, C_k>0$ for $k=0, 1$ such that 
\begin{equation*}
|\partial^k q| \le C_kf^{-1-k-\rho}.
\end{equation*}
\end{cond}
Under Condition~\ref{cond:short-range} 
it follows by the Faris--Lavine theorem (cf. \cite{RS}) that
$H_\alpha$ is essentially self-adjoint on $C_0^\infty(\R^d)$.
We denote its self-adjoint extension by the same letter for simplicity.

Next we introduce the Agmon--H\"ormander spaces associated with the function $f$.
We let $F(S)$ be the sharp characteristic function 
of a general subset $S\subseteq \R^d$,
and set
\begin{equation*}
F_n = F\bigl(\bigl\{ x\in \R^d\,\big|\,2^n\le f(x)<2^{n+1}\bigr\} \bigr)
\quad \text{for }n\in\N_0.
\end{equation*}
Then define the Agmon--H\"ormander spaces $\vB$, $\vB^*$ and $\vB^*_0$ as 
\begin{align*}
\begin{split}
\vB &=
\Bigl\{\psi\in L^2_{\mathrm{loc}}(\R^d)\,\Big|\, 
\|\psi\|_\vB:=\sum_{n\in\N_0}2^{n/2}\|F_n\psi\|_{L^2}<\infty\Bigr\},
\\
\vB^* &=
\Bigl\{\psi\in L^2_{\mathrm{loc}}(\R^d)\,\Big|\, 
\|\psi\|_{\vB^*}\!:=\sup_{n\in\N_0}2^{-n/2}\|F_n\psi\|_{L^2}<\infty\Bigr\},
\\
\vB^*_0 &=
\Bigl\{\psi\in \vB^*\,\Big|\, \lim_{n\to\infty}2^{-n/2}\|F_n\psi\|_{L^2}=0 \Bigr\}.
\end{split}
\end{align*}
Note that $\vB$ is a Banach space with respect to the norm $\|\cdot\|_\vB$, 
and $\vB^*$ and $\vB_0^*$ are Banach spaces 
with respect to the same norm $\|\cdot\|_{\vB^*}$.
Note also that, 
if we introduce the \emph{$f$-weighted $L^2$-spaces of order $s\in\R$} as 
$$L^2_s=f^{-s}L^2,$$
for any $s>1/2$ the following inclusion relations hold:
\begin{align*}
L^2_s \subsetneq \vB \subsetneq L^2_{1/2} \subsetneq L^2 
\subsetneq L^2_{-1/2} \subsetneq \vB_0^* \subsetneq \vB^* \subsetneq L^2_{-s}.
\end{align*}

We introduce differential operators $\partial^r$ and $\partial^f$ as
\begin{equation*}
\partial^r=(\partial r)\partial, \quad
\partial^f=(\partial f)\partial=r^{-\alpha/2}(\partial r)\partial,
\end{equation*}
respectively, and then we define a `conjugate operator' $A$ as
\begin{equation}\label{eq:defA}
A=\mathop{\mathrm{Re}}p^f=\tfrac12\left((p^f)^*+p^f \right), \quad p^f=-i\partial^f.
\end{equation}
We note that (cf. \cite{Ita1}) 
$A$ is self-adjoint operator with domain $\{\psi\in L^2\,|\,A\psi\in L^2\}$, 
and has expressions
$$
A=p^f-\tfrac{i}2(\Delta f)=(p^f)^*+\tfrac{i}2(\Delta f).
$$
We also note $A$ is different from the standard conjugate operator, cf. \cite{BCHM, Mo}. 
In fact, the commutator $[H, iA]$ has only weak positivity decaying at infinity, 
see \cite{Ita1}. 
We denote the resolvent of $H_\alpha$ for $z\in \C\setminus\R$ by $R(z)$, i.e., 
\begin{equation*}
R(z)=(H_\alpha-z)^{-1}.
\end{equation*}

Let us introduce the function $\theta\in C^\infty(\R^d)$ by 
\begin{equation}\label{eq:theta-eikonal}
\theta(\lambda, x)=r^{1+\alpha/2}/(1+\alpha/2)+\lambda f.
\end{equation}
Note that the function $\theta$ is an approximate solution to the eikonal equation
\begin{equation*}
\tfrac12|\tfrac{\partial\theta}{\partial x}(\lambda, x)|^2-\tfrac12|x|^\alpha+q-\lambda=0, 
\end{equation*}
in the sense that for $2/3<\alpha<2$ the quantity of the left-hand side tends to $0$ 
faster than $f^{-1}$ as $f\to\infty$.
More precisely, the function $\theta$ satisfies 
\begin{equation}\label{eq:2003151426}
\tfrac12|\tfrac{\partial\theta}{\partial x}(\lambda, x)|^2-\tfrac12|x|^\alpha+q-\lambda 
= \vO(f^{-1-\min\{\rho, (3\alpha/2-1)/(1-\alpha/2)\}}).
\end{equation}

\begin{remark}
(1) When $\alpha=1,2$, 
the functions $\theta_1, \theta_2\in C^\infty(\R^d)$ satisfying 
\begin{align*}
\theta_1(\lambda, x) &= \tfrac23(r+2\lambda)^{3/2} \quad \text{for } r>1-2\lambda, 
\\
\theta_2(\lambda, x) &= \tfrac12r(r^2+2\lambda)^{1/2}+\lambda\log\{\tfrac{r}2+\tfrac12(r^2+2\lambda)^{1/2}\} \quad \text{for } r^2>1-2\lambda,  
\end{align*} 
respectively, 
solve the eikonal equation of the \emph{free} case for sufficiently large $r$. 
In particular, their leading terms coincide with $\theta(\lambda, x)$.

(2) We constructed $\theta$ by the following simple approximation. 
\begin{align*}
\theta(\lambda, x) 
\sim 
\int(r^\alpha+2\lambda)^{1/2}(\partial r) \,{\rm d}x 
\sim 
\int \bigl(r^{\alpha/2}+\lambda r^{-\alpha/2}+\vO(r^{-3\alpha/2})\bigr)(\partial r) \,{\rm d}x.
\end{align*}
Thus by adding some lower order terms to $\theta(\lambda, x)$, 
we can improve the order of the right-hand side of \eqref{eq:2003151426}.
\end{remark}

In the following we assume that $2/3<\alpha<2$. 
However our results hold for all $\alpha\in (0, 2)$ by retaking $\theta$ appropriately, 
as stated in the above remark.


\subsection{Results of our previous papers}\label{sec:Results of previous papers}

Before stating our main results, let us recall several results of \cite{Ita1, Ita2}. 
Because we use the radiation condition bounds for limiting resolvents 
of the forms of \cite{Ita2} and 
Sommerfeld's uniqueness theorem to construct 
the stationary wave operators and the scattering matrix, respectively. 
However, as for the radiation condition bounds we need slightly stronger estimates 
than those of \cite{Ita2} even for the short--range case. 
Thus we need to refine the results except for Rellich's theorem 
and limiting absorption principle bounds. 

In this section we state improved results. 
We prove only Theorem~\ref{thm:rc-complex} stated below, since
if we get Theorem~\ref{thm:rc-complex} other result can be proved 
by quite similar way to \cite{AIIS, Ita2}. 
The proof of Theorem~\ref{thm:rc-complex} is given 
in Appendix~\ref{Appen:proof-of-thm}. 

The first result is the absence of $\vB_0^*$-eigenfunctions, 
which is called Rellich's theorem. 
Since the condition on $q$ of this paper is stronger than that of \cite{Ita1}, 
we have the following theorem. 
\begin{thm}\label{thm:rellich}
Let $\lambda\in\R$. Suppose a function $\phi\in\vB_0^*$ satisfies 
\begin{equation*}
(H_\alpha-\lambda)\phi=0
\end{equation*}
in the distributional sense. 
Then $\phi=0$ on $\R^d$.
\end{thm}

We set 
\begin{equation}\label{eq:f-ell}
\ell_{jk}=|\partial f|^2\delta_{jk}-(\partial_jf)(\partial_kf),
\end{equation}
where $\delta_{jk}$ is Kronecker's delta.
For any compact interval $I\subseteq\R$ we introduce
\begin{equation*}
I_\pm=\{z=\lambda\pm i\Gamma\,|\,\lambda\in I,\ \Gamma\in(0,1)\},
\end{equation*}
respectively.
We also use the notation $\langle T\rangle_\psi=\langle \psi, T\psi\rangle$ 
for a general linear operator $T$.
The following limiting absorption principle bounds (LAP bounds) 
for the resolvent $R(z)$ also hold in the setting of this paper.
\begin{thm}\label{thm:lap-bound}
There exists $C>0$ such that for any $\psi\in\vB$ and $z\in I_\pm$
\begin{equation*}
\|R(z)\psi\|_{\vB^*} 
+ \|p^fR(z)\psi\|_{\vB^*} 
+ \langle p_jf^{-1}\ell_{jk}p_k\rangle_{R(z)\psi}^{1/2} 
+ \|r^{-\alpha} p^2R(z)\psi\|_{\vB^*} 
\le 
C\|\psi\|_\vB.
\end{equation*}
\end{thm}

Using the function $\chi$ of \eqref{eq:chidef},
we define \emph{smooth cut-off functions} 
$\chi_m,\bar\chi_m,\chi_{m,n}\in C^\infty(\R^d)$ for $m,n\in\N_0$ as 
\begin{equation}\label{eq:chimn}
\chi_m=\chi(f/2^m), \quad \bar \chi_m=1-\chi_m, \quad \chi_{m,n}=\bar\chi_m\chi_n.
\end{equation}
We choose and fix large $m\in\N$ so that 
on $\mathop{\mathrm{supp}}\bar\chi_m$
\begin{equation*}
2\mathop{\mathrm{Re}}z-2q_0+r^\alpha>1, \quad r=|x|,
\end{equation*}
where $z\in I_\pm$ and 
\begin{equation*}
q_0 = q 
+ \tfrac18r^\alpha(\Delta f)^2 
+ \tfrac{\alpha}4r^{\alpha/2-1}(\Delta f) 
+ \tfrac14r^\alpha(\partial^f\Delta f) 
- \tfrac\alpha4r^{-2}.
\end{equation*}
We set an asymptotic complex phase $a$ by 
\begin{equation}\label{eq:improved-phase}
a = a_z 
= \bar\chi_m\Bigl[ r^{-\alpha/2}\sqrt{2(z-q_0)+r^\alpha} 
\pm \tfrac{i\alpha}2r^{-\alpha/2-1} 
\mp \tfrac{i\alpha}2\tfrac{z-q_0}{2(z-q_0)+r^\alpha}r^{-\alpha/2-1} \Bigr] 
\end{equation}
for $z\in I_\pm$.
Here we choose the branch of square root as $\mathop{\mathrm{Re}}\sqrt{s}>0$ 
for $s \in \C\setminus(-\infty,0]$.
Let
\begin{equation*}
\beta_c=\min\{\rho+1/(1-\alpha/2), 1+\alpha/(1-\alpha/2)\}.
\end{equation*}
Then we have refined radiation condition bounds for complex spectral parameters. 
\begin{thm}\label{thm:rc-complex}
For all $\beta\in[0,\beta_c)$, there exists $C>0$ such that 
for any $\psi\in f^{-\beta}\vB$ and $z\in I_\pm$
\begin{equation}\label{eq:2001091149}
\|f^\beta(A\mp a)R(z)\psi\|_{\vB^*} 
+ \langle p_jf^{2\beta-1}\ell_{jk}p_k\rangle_{R(z)\psi}^{1/2} 
\le 
C\|f^\beta\psi\|_\vB,
\end{equation}
respectively.
\end{thm}

Let us state several applications of Theorem~\ref{thm:lap-bound} 
and Theorem~\ref{thm:rc-complex}. 
The first one is the limiting absorption principle. 
\begin{corollary}\label{cor:lap-refined}
For any $s>1/2$ and $\omega\in(0,\beta_c)\cap(0,\min\{s-1/2,1\}]$ there exists $C>0$ 
such that for any $z,z'\in I_+$ or $z,z'\in I_-$
\begin{equation*}
\begin{split}
\|R(z)-R(z')\|_{\vL(\vH_s,\vH_{-s})} &\le C|z-z'|^\omega, \\
\|r^{-\alpha/2}p\bigl\{R(z)-R(z')\bigr\}\|_{\vL(\vH_s,\vH_{-s})} &\le C|z-z'|^\omega.
\end{split}
\end{equation*}
In particular, for any $\lambda\in\R$, there exist uniform limits 
\begin{equation*}
\lim_{(0,1)\ni \Gamma\searrow 0}R(\lambda\pm i\Gamma), 
\quad 
\lim_{(0,1)\ni \Gamma\searrow 0}r^{-\alpha/2}pR(\lambda\pm i\Gamma),
\end{equation*}
in the norm topology of $\vL(\vH_s,\vH_{-s})$.
We denote these limits by 
$R(\lambda\pm i0), r^{-\alpha/2}pR(\lambda\pm i0)$, respectively.
These limiting resolvents belong to $\vL(\vB,\vB^*)$.
\end{corollary}

The second one is the radiation condition bounds for real spectral parameters, 
which follows from Theorem~\ref{thm:rc-complex} and Corollary~\ref{cor:lap-refined}.
We set 
\begin{equation*}
a_\pm := \lim_{I_\pm\ni z\to \lambda\pm i0} a_z, \quad \lambda\in I.
\end{equation*}

\begin{corollary}\label{cor:rc-real}
Let $\lambda\in I$. 
Then for all $\beta\in[0,\beta_c)$, there exists $C>0$ such that 
for any $\psi\in f^{-\beta}\vB$ 
\begin{equation*}
\|f^\beta(A\mp a_\pm)R(\lambda\pm i0)\psi\|_{\vB^*} 
+ \langle p_jf^{2\beta-1}\ell_{jk}p_k\rangle_{R(\lambda\pm i0)\psi}^{1/2} 
\le 
C\|f^\beta\psi\|_\vB,
\end{equation*}
respectively.
\end{corollary}

The last one is Sommerfeld's uniqueness theorem.
\begin{corollary}\label{cor:uniqueness-result-refine}
Let $\lambda\in\R, \phi\in f^\beta\vB^*$ 
and $\psi\in f^{-\beta}\vB$ with $\beta\in [0,\beta_c)$.
Then $\phi=R(\lambda\pm i0)\psi$ hold if and only if both of the following conditions hold:
\begin{itemize}
\item[(i)] $(H_\alpha-\lambda)\phi=\psi$ in the distributional sense.
\item[(ii)] $(A\mp a_\pm)\phi\in f^{-\beta}\vB_0^*$, 
\end{itemize}
respectively. 
\end{corollary}


\subsection{Main results}\label{sec:Main results}

We introduce the operators $\scF^\pm(\lambda, f)$ 
which map from $C_0^\infty(\R^d)$ to $L^2(\mathbb S^{d-1})$ by 
\begin{equation*}
\begin{split}
&\quad\hspace{-3mm}
\bigl( \scF^\pm(\lambda, f)\psi\bigr)(\omega) 
\\&= 
\tfrac1{\sqrt{2\pi}}
\exp\Bigl\{\pm\tfrac{\pi i}4\Bigl(\tfrac{d-\alpha/2-3}{1+\alpha/2}\Bigr)\!\Bigr\}
r^{(d+\alpha/2-1)/2}e^{\mp i\theta(\lambda, \cdot\omega)}
\bigl(R(\lambda\pm i0)\psi\bigr)(\pm\cdot\omega), 
\end{split}
\end{equation*}
respectively, 
where $\psi\in C_0^\infty(\R^d)$ and $\omega\in \mathbb S^{d-1}$.
\begin{thm}\label{thm:existence-of-wave-matrix}
$\scF^\pm(\lambda, f)$ are bounded operators 
from $C_0^\infty(\R^d)\ (\subseteq\vB)$ to $L^2(\mathbb S^{d-1})$, 
and for any $\psi\in C_0^\infty(\R^d)$ there exist limits 
\begin{equation}\label{eq:2003161503}
\lim_{f\to\infty}\scF^\pm(\lambda, f)\psi 
\equiv 
\scF^\pm(\lambda)\psi \quad \text{in }\ L^2(\mathbb S^{d-1}). 
\end{equation}
Moreover it holds that 
\begin{equation}\label{eq:2003161512}
\tfrac1{2\pi i}\langle R(\lambda+i0)\psi - R(\lambda-i0)\psi, \psi\rangle 
= 
\|\scF^\pm(\lambda)\psi\|^2_{L^2(\mathbb S^{d-1})}.
\end{equation}
\end{thm}

We note, by \eqref{eq:2003161512}, 
the operators $\scF^\pm(\lambda)$ are extended to bounded operators 
from $\vB$ into $L^2(\mathbb S^{d-1})$, 
and satisfy $\|\scF^+(\lambda)\psi\|=\|\scF^-(\lambda)\psi\|$ for any $\psi\in\vB$.
We also note that $\scF^\pm(\lambda)$ are continuous in $\lambda\in\R$. 
This follows from the continuity of $R(\lambda\pm i0)$ in $\lambda\in\R$, 
\eqref{eq:2003161512} and \eqref{eq:2003301812} stated below.

We introduce the spaces
\begin{equation*}
\vH=L^2(\R^d),\quad 
\widetilde\vH=L^2\bigl(\R, {\rm d}\lambda; L^2(\mathbb S^{d-1})\bigr), 
\end{equation*}
and define the operators $\vF^\pm:\vB\to C\bigl(\R; L^2(\mathbb S^{d-1})\bigr)$ 
as 
\begin{equation}\label{eq:2005270050}
\bigl(\vF^\pm\psi\bigr)(\lambda) = \scF^\pm(\lambda)\psi,\quad \psi\in\vB, 
\end{equation}
respectively. 
\begin{proposition}\label{prop:2005302258}
The operators $\vF^\pm$ defined as mappings 
$\vB\to C\bigl(\R; L^2(\mathbb S^{d-1})\bigr)$ by \eqref{eq:2005270050} 
extend uniquely to isometries $\vH\to\widetilde\vH$. 
These operators satisfy $\vF^\pm H_\alpha\subseteq M_\lambda\vF^\pm$. 
\end{proposition}

We call the operator $\vF^\pm: \vH\to\widetilde\vH$ stationary wave operator. 
Existence of the stationary wave operators follows 
from Proposition~\ref{prop:2005302258}. 
Since $\mathop{\mathrm{Ran}}\scF^\pm(\lambda)$ are dense 
in $L^2(\mathbb S^{d-1})$ (see \eqref{eq:2004021457} below), 
by the density argument we can see 
that $\vF^\pm$ are surjection, see e.g. \cite{ACH}. 
Therefore we obtain the completeness of the stationary wave operators. 
\begin{thm}\label{thm:wave-op-unitary}
The operators $\vF^\pm$ are unitary, and satisfy 
\begin{equation*}
\vF^\pm H_\alpha= M_\lambda\vF^\pm.
\end{equation*}
\end{thm}

Let us introduce the functions $\phi_\lambda^\pm[v]$ 
for $v\in L^2(\mathbb S^{d-1})$ by 
\begin{equation}\label{eq:2004021344}
\phi_\lambda^\pm[v](f, \omega) 
= \tfrac1{\sqrt{2\pi}}
\exp\Bigl\{\mp \tfrac{\pi i}4\Bigl(\tfrac{d+\alpha/2-1}{1+\alpha/2}\Bigr)\!\Bigr\}
r^{-(d+\alpha/2-1)/2}e^{\pm i\theta(\lambda, x)}v(\pm\omega), 
\end{equation}
respectively. 
We may call these functions outgoing/incoming approximate generalized eigenfunctions. 
In fact for $v\in C^\infty(\mathbb S^{d-1})$ we can see that 
\begin{equation}\label{eq:2004021416}
\psi_\lambda^\pm[v]:=(H_\alpha-\lambda)\phi_\lambda^\pm[v]\in\vB, 
\end{equation}
cf. \eqref{eq:2003191524}.
The adjoints of $\scF^\pm(\lambda)$: 
\begin{equation*}
\scF^\pm(\lambda)^*\in\vL(L^2(\mathbb S^{d-1}), \vB^*), 
\end{equation*}
which are called the \emph{stationary wave matrices}, 
are characterized by $\phi_\lambda^\pm$ and $\psi_\lambda^\pm$ as follows. 
\begin{proposition}\label{prop:2004031943}
Let $v\in C^\infty(\mathbb S^{d-1})$, 
and let $\phi_\lambda^\pm[v]$ and $\psi_\lambda^\pm[v]$ be given by 
\eqref{eq:2004021344} and \eqref{eq:2004021416}, respectively. 
Then 
\begin{equation}\label{eq:2004032011}
\scF^\pm(\lambda)^*v 
= \phi_\lambda^\pm[v] - R(\lambda\mp i0)\psi_\lambda^\pm[v] \ (\in\vB^*), 
\end{equation}
respectively.
\end{proposition}

$\scF^\pm(\lambda)^*$ are also called \emph{eigenoperators}.
In fact, by Proposition~\ref{prop:2004031943} and the density argument 
we can see that 
\begin{equation*}
(H_\alpha-\lambda)\scF^\pm(\lambda)^*v=0\quad 
\text{for any}\ v\in L^2(\mathbb S^{d-1}).
\end{equation*}


By Sommerfeld's uniqueness theorem 
stated as Corollary~\ref{cor:uniqueness-result-refine}, 
we have 
\begin{equation*}
\phi_\lambda^\pm[v] - R(\lambda\pm i0)\psi_\lambda^\pm[v]
= 0 \quad \text{for }\ v\in C^\infty(\mathbb S^{d-1}).
\end{equation*}
We can deduce from this equality that 
\begin{equation}\label{eq:2004021457}
v = \pm 2\pi i\scF^\pm(\lambda)\psi_\lambda^\pm[v] 
\quad \text{for }\ v\in C^\infty(\mathbb S^{d-1}), 
\end{equation}
and then we have 
\begin{equation*}
C^\infty(\mathbb S^{d-1}) \subseteq \mathop{\mathrm{Ran}}\scF^\pm(\lambda) 
\subseteq L^2(\mathbb S^{d-1}).
\end{equation*}
Therefore we can define the \emph{scattering matrix} $S(\lambda)$ 
as satisfying for $\psi\in\vB$
\begin{equation}\label{eq:2008111454}
\scF^+(\lambda)\psi = S(\lambda)\scF^-(\lambda)\psi.
\end{equation}
Then by Theorem~\ref{thm:existence-of-wave-matrix} 
we can see that the scattering matrix is extended to an unitary operator. 
\begin{proposition}\label{prop:2008101448}
$S(\lambda)$ defined by \eqref{eq:2008111454} is extended to a unitary operator 
on $L^2(\mathbb S^{d-1})$ and is strongly continuous in $\lambda\in\R$. 
\end{proposition}

Finally, we obtain a characterization of the $\vB^*$-eigenfunctions 
in terms of $\phi_\lambda^\pm$ similar to \cite{IS2}. 
Let us introduce the set of minimal generalized eigenfunctions.
\begin{equation*}
\mathcal E_\lambda
:=\{\phi\in\vB^*\,|\, (H_\alpha-\lambda)\phi=0\ \, \text{in the distributional sense.}\,\}.
\end{equation*}

\begin{thm}\label{thm:characterize-B^*-ef}
For any fixed $\lambda\in\R$ the following assertions hold.
\begin{enumerate}
\item[(i)] For any one of $\xi_\pm\in L^2(\mathbb S^{d-1})$ 
or $\phi\in \mathcal E_\lambda$ 
the two other quantities in $\{\xi_+, \xi_-, \phi\}$ uniquely exist such that 
\begin{equation}\label{eq:2005232312}
\phi-\phi_\lambda^+[\xi_+]-\phi_\lambda^-[\xi_-]\in \vB_0^*.
\end{equation}
\item[(ii)] For the quantities $\{\xi_+, \xi_-, \phi\}$ satisfying \eqref{eq:2005232312}, 
the following relations hold.
\begin{align}\label{eq:2005241246}
\phi &= \scF^\pm(\lambda)^*\xi_\pm,\qquad 
\xi_+=S(\lambda)\xi_-, \\
\xi_\pm &= \pm\dfrac12c_\pm \lim_{R\to\infty}
\dfrac1R\int_R^{2R} r^{(d+\alpha/2-1)/2}e^{\mp i\theta}(A\pm a_0)\phi\,{\rm d}f, 
\label{eq:2005241247}
\end{align}
where $c_\pm
=\sqrt{2\pi}\exp\{\pm\tfrac{\pi i}4(\tfrac{d+\alpha/2-1}{1+\alpha/2})\}$ 
and $a_0=r^{-\alpha/2}\sqrt{2\lambda-2q_0+r^\alpha}$.
In particular the wave matrices 
$\scF^\pm(\lambda)^*$ give one-to-one correspondences between 
the spaces $L^2(\mathbb S^{d-1})$ and $\vE_\lambda$. 
\item[(iii)] The operators $\scF^\pm(\lambda):\vB\to L^2(\mathbb S^{d-1})$ 
are surjections.
\end{enumerate}
\end{thm}

There are many literature on scattering theory for the Laplacian with decaying potentials. 
We refer e.g. \cite{BM, GY, Ho, Ike, II, Iso, IS2, M, MZ, Sa}. 
However there seems to be no literature on stationary scattering problem 
for repulsive Hamiltonians even for short--range perturbation, 
although time-dependent scattering problem for that is well studied 
cf. e.g. \cite{BCHM, Ishi, N}. 
In this sense, our results are new. 
Moreover since we use Agmon--H\"ormander spaces, which are used only 
for the Laplacian with decaying potentials, so far, cf. \cite{AH, JP}, 
our results have sharp form. 

To prove our main results we use the schemes of \cite{GY, Iso, IS2} as mentioned above. 
Since considered Hamiltonians in this paper are different from theirs, 
we can not apply their schemes directly. 
Thus we improve that by using escape function \eqref{eq:escape-function} 
and an approximate solution \eqref{eq:theta-eikonal} to the eikonal equation. 

We have already obtained similar results to \cite{IS1, Ita1, Ita2} in \cite{AIIS} 
for the Stark Hamiltonians. 
Thus by considering the results of \cite{IS2} and this paper, 
we can expect that stationary scattering theory can be established 
for the Stark Hamiltonians.


\section{Wave operator}\label{sec:Wave operator}

In this section we discuss on the stationary wave operators $\vF^\pm$. 
In Section~\ref{sec:Stationary state} 
we prove Theorems~\ref{thm:existence-of-wave-matrix} 
by employing Isozaki's approach, cf. \cite{Iso, GY}.
Proposition~\ref{prop:2005302258} will be proved 
in Section~\ref{sec:Stationary wave operators}.


\subsection{Stationary state}\label{sec:Stationary state}

To prove Theorem~\ref{thm:existence-of-wave-matrix} 
let us introduce the following lemmas. 
\begin{lemma}\label{lem:2003031452}
Let $\psi\in C_0^\infty(\R^d)$ and $\phi=R(\lambda\pm i0)\psi$.
Then
\begin{equation}\label{eq:2003031457}
\langle R(\lambda+i0)\psi-R(\lambda-i0)\psi, \psi\rangle
= i\lim_{\tilde f\to\infty}\int_{f(x)=\tilde f}|\phi|^2 \,{\rm d}S_{f=\tilde f}.
\end{equation}
\end{lemma}

\begin{lemma}\label{lem:2003051708}
Let $\psi\in C_0^\infty(\R^d)$. 
Then there exists a weak limit
\begin{equation}\label{eq:2003051710}
\mathop{\mathrm{w\text{-}lim}}_{f\to\infty}\scF^\pm(\lambda, f)\psi 
\equiv \scF^\pm(\lambda)\psi 
\quad \text{ in }\  L^2(\mathbb S^{d-1}).
\end{equation}
\end{lemma}

\begin{lemma}\label{lem:2023181624}
Let $\psi\in C_0^\infty(\R^d)$.
Then there exists a sequence $\{f_n\}_{n\in\N}$ 
satisfying $f_n\to\infty$ as $n\to\infty$ 
such that $\scF^\pm(\lambda, f_n)\psi$ tends to $\scF^\pm(\lambda)\psi$ 
in $L^2(\mathbb S^{d-1})$ as $n\to\infty$.
\end{lemma}

By using the function \eqref{eq:theta-eikonal}, 
we define the differential operators $\scD^\pm_j, \scD^\pm, \scD^\pm_f$ 
and $\scD^\pm_\omega$ by 
\begin{equation*}
\begin{split}
\scD^\pm_j &= \partial_j 
+\tfrac{d+\alpha/2-1}{2r}\tfrac{x_j}{r}\mp i\tfrac{\partial\theta}{\partial x_j}(\lambda, x), 
\quad j=1,\ldots,d, 
\\
\scD^\pm &= (\scD^\pm_1, \ldots, \scD^\pm_d), 
\\
\scD^\pm_f &= r^{-\alpha/2}\tfrac{x_j}r\!\cdot\!\scD^\pm_j, 
\\ 
\scD^\pm_\omega &= \partial_x-\tfrac{x}{r}\partial_{|x|} 
\mp i\bigl(\tfrac{\partial\theta}{\partial x}(\lambda, x)-\tfrac{x}{r}\tfrac{\partial\theta}{\partial |x|}(\lambda, x)\bigr),
\end{split}
\end{equation*}
respectively.
On the region $\{x\in \R^d\,|\,r(x)\ge2\}$, $\scD^\pm_f$ is expressed as  
\begin{equation*}
\scD^\pm_f = r^{-\alpha}\partial_f + \tfrac{d+\alpha/2-1}{2}r^{-\alpha/2-1}\mp ir^{-\alpha/2}\tfrac{\partial\theta}{\partial |x|}(\lambda, x),
\end{equation*}
and these operators satisfy the identity
\begin{equation}\label{eq:2003021706}
\sum_{j=1}^d(\scD^\pm_j)^*\scD^\pm_j 
= (r^{\alpha/2}\scD^\pm_f)^*r^{\alpha/2}\scD^\pm_f 
+ (\scD^\pm_\omega)^*\scD^\pm_\omega 
= (\scD^\pm_f)^*r^\alpha\scD^\pm_f + (\scD^\pm_\omega)^*\scD^\pm_\omega,
\end{equation}
where 
\begin{align*}
(\scD^\pm_j)^* &= -\partial_j + \tfrac{d+\alpha/2-1}{2r}\tfrac{x_j}{r} 
\pm i\tfrac{\partial\theta}{\partial x_j}(\lambda, x), 
\\ 
(\scD^\pm_f)^* &= -\partial_fr^{-\alpha} - \tfrac{d+\alpha/2-1}{2}r^{-\alpha/2-1} 
\pm ir^{-\alpha/2}\tfrac{\partial\theta}{\partial |x|}(\lambda, x),
\\ 
(\scD^\pm_\omega)^* &= -\partial_x+\partial_{|x|}\tfrac{x}{r}+(d-1)r^{-1}\tfrac{x}{r} 
\pm i\bigl(\tfrac{\partial\theta}{\partial x}(\lambda, x)-\tfrac{x}{r}\tfrac{\partial\theta}{\partial |x|}(\lambda, x)\bigr).
\end{align*}

We also note that 
$-i\scD^\pm_f=A\mp a_\pm+\vO(r^{-\alpha}f^{-\min\{1+\rho, \alpha/(1-\alpha/2)\}})$, 
and thus by noting $2/3<\alpha<2$ 
it follows from Corollary~\ref{cor:rc-real} 
that for any $\psi\in\vB$ and $\beta\in[0,\beta_c)$ 
\begin{equation}\label{eq:2003181148}
f^\beta\scD^\pm_fR(\lambda\pm i0)\psi\in \vB^*.
\end{equation}

We prove Lemmas~\ref{lem:2003031452}--\ref{lem:2023181624} 
and Theorem~\ref{thm:existence-of-wave-matrix} only for the upper sign. 
Thus in the following we consider only for 
$\scD^+_j, \scD^+, \scD^+_f, \scD^+_\omega, \scF^+(\lambda, f)$ and $\scF^+(\lambda)$, 
and then, for notational simplicity, we omit the superscript from these operators 
such as $\scD_j$.

First we see the following property of $\phi=R(\lambda\pm i0)\psi$ for 
$\psi\in C_0^\infty(\R^d)$.
\begin{lemma}\label{lem:2003171406}
Let $\psi\in C_0^\infty(\R^d)$ and $\phi=R(\lambda\pm i0)\psi$.
Then
\begin{equation*}
\lim_{\tilde f\to\infty} 
\int_{f(x)=\tilde f}(\scD_f\phi)\bar\phi - \phi(\overline{\scD_f\phi}) \,{\rm d}S_{f=\tilde f} 
= 0.
\end{equation*}
\end{lemma}

\begin{proof}
Let $\psi$ and $\phi$ be as in the assertion, and $\tilde f\ge2$.
Noting the expressions of $\scD_f$ and $\scD_f^*$ we can compute as 
\begin{align*}
&\quad\hspace{-3mm}
\tfrac{\rm d}{{\rm d}\tilde f}\int_{f(x)=\tilde f}(\scD_f\phi)\bar\phi \,{\rm d}S_{f=\tilde f} 
\\&= 
\int_{f(x)=\tilde f} r^\alpha|\scD_f\phi|^2\,{\rm d}S_{f=\tilde f} 
- (d+\alpha/2-1)\int_{f(x)=\tilde f}r^{\alpha/2-1}(\scD_f\phi)\bar\phi\,{\rm d}S_{f=\tilde f}  
\\&\phantom{{}={}}{}
+ \int_{f(x)=\tilde f}(-\scD_f^*r^\alpha\scD_f\phi)\bar\phi\,{\rm d}S_{f=\tilde f}. 
\end{align*}
We can see by Corollary~\ref{cor:rc-real} that 
the first and the second terms belong to $L^1((2,\infty))$.
As for the third term, 
let us further compute the factor $-\scD_f^*r^\alpha\scD_f\phi$.
By a straightforward calculation we have
\begin{equation}
\begin{split}\label{eq:2004031644}
&\quad\hspace{-3mm}
\sum_{j=1}^d\scD_j^*\scD_j\phi 
- 2i\sum_{j=1}^d\tfrac{\partial\theta}{\partial x_j}(\lambda, x)\scD_j\phi 
\\&= 
2\psi 
- \Bigl\{\bigl|\tfrac{\partial\theta}{\partial x}(\lambda, x)\bigr|^2-r^\alpha+2q-2\lambda 
+ \tfrac{(d+\alpha/2-1)(d-\alpha/2-3)}{4r^2} \Bigr\}\phi 
\\&\phantom{{}={}}{}
+ i\Bigl\{(\Delta\theta)(\lambda, x)-(d+\alpha/2-1)r^{-1}\tfrac{\partial\theta}{\partial |x|}(\lambda, x) \Bigr\}\phi. 
\end{split}
\end{equation}
Therefore by using \eqref{eq:2003021706}, \eqref{eq:2004031644} and 
the relation $\scD=\scD_\omega+\tfrac{x}rr^{\alpha/2}\scD_f$, we can obtain 
\begin{align*}
-\scD_f^*r^\alpha\scD_f\phi 
&= 
\scD_\omega^*\scD_\omega\phi 
- 2ir^\alpha\scD_f\phi 
- 2i\lambda\scD_f\phi 
- 2\psi 
\\&\phantom{{}={}}{}
+ \Bigl\{\bigl|\tfrac{\partial\theta}{\partial x}(\lambda, x)\bigr|^2 
- r^\alpha + 2q - 2\lambda+\tfrac{(d+\alpha/2-1)(d-\alpha/2-3)}{4r^2} \Bigr\}\phi 
\\&\phantom{{}={}}{}
- i\Bigl\{(\Delta\theta)(\lambda, x) 
- (d+\alpha/2-1)r^{-1}\tfrac{\partial\theta}{\partial |x|}(\lambda, x) \Bigr\}\phi. 
\end{align*}
Hence by Theorem~\ref{thm:lap-bound}, Corollary~\ref{cor:rc-real}, 
\eqref{eq:theta-eikonal}, \eqref{eq:2003151426} and \eqref{eq:2003181148} 
we have 
\begin{equation}
\begin{split}\label{eq:2003181202}
&\quad\hspace{-3mm}
\tfrac{\rm d}{{\rm d}\tilde f}\int_{f(x)=\tilde f}(\scD_f\phi)\bar\phi \,{\rm d}S_{f=\tilde f} 
\\&= 
\int_{f(x)=\tilde f}(\scD_\omega^*\scD_\omega\phi)\bar\phi \,{\rm d}S_{f=\tilde f} 
- 2i\tilde r^\alpha\int_{f(x)=\tilde f}(\scD_f\phi)\bar\phi \,{\rm d}S_{f=\tilde f} 
+ G_1(\tilde f),
\end{split}
\end{equation}
where $G_1$ is a certain function which satisfying $\int_2^\infty|G_1(s)|\,{\rm d}s<\infty$, 
and $\tilde r\ge2$ is a solution of 
\begin{equation*}
\tilde f = (\tilde r^{1-\alpha/2}-1)/(1-\alpha/2)+1.
\end{equation*}
Since we can see by integration by parts
\begin{equation*}
\int_{f(x)=\tilde f}(\scD_\omega^*\scD_\omega\phi)\bar\phi \,{\rm d}S_{f=\tilde f} 
= 
\int_{f(x)=\tilde f}|\scD_\omega\phi|^2 \,{\rm d}S_{f=\tilde f}, 
\end{equation*}
we obtain from \eqref{eq:2003181202} that 
\begin{equation*}
\begin{split}
&\quad\hspace{-3mm}
\tfrac{\rm d}{{\rm d}\tilde f}\int_{f(x)=\tilde f}(\scD_f\phi)\bar\phi - \phi(\overline{\scD_f\phi}) \,{\rm d}S_{f=\tilde f} 
\\&= 
- 2i\tilde r^\alpha\int_{f(x)=r_n}(\scD_fu)\bar u - u(\overline{\scD_fu}) {\rm d}S_{f(x)=r_n} 
+ G(\tilde f), 
\end{split}
\end{equation*}
where $G=G_1-\overline{G_1}$. 
Now, if we let 
\begin{equation*}
u(\tilde f) = 
e^{2i(\tilde r^{1+\alpha/2})/(1+\alpha/2)}v(\tilde f),\quad 
v(\tilde f) = 
\int_{f(x)=\tilde f}(\scD_f\phi)\bar\phi - \phi(\overline{\scD_f\phi}) \,{\rm d}S_{f=\tilde f}, 
\end{equation*}
we have 
\begin{equation*}
\tfrac{\rm d}{{\rm d}\tilde f}u(\tilde f) 
= e^{2i(\tilde r^{1+\alpha/2})/(1+\alpha/2)}G(\tilde f).
\end{equation*}
The solution of this differential equation is given by 
\begin{equation*}
u(\tilde f) 
= u(2)+\int_2^{\tilde f} e^{2i(r^{1+\alpha/2})/(1+\alpha/2)}G(f) \,{\rm d}f. 
\end{equation*}
Since $G(f)\in L^1((2,\infty))$, 
there exists a limit $\lim_{\tilde f\to\infty}u(\tilde f)$. 
On the other hand, by Theorem~\ref{thm:lap-bound}, Corollary~\ref{cor:rc-real} 
and the Cauchy-Schwarz inequality, we have 
\begin{equation*}
\int_2^\infty|v(\tilde f)|\,{\rm d}\tilde f < \infty.
\end{equation*}
This implies that 
$\displaystyle\liminf_{\tilde f\to\infty}|v(\tilde f)| 
= 
\liminf_{\tilde f\to\infty}|u(\tilde f)| =0$.
Therefore we have $\displaystyle\lim_{\tilde f\to\infty}u(\tilde f)=0$, 
or 
$\displaystyle\lim_{\tilde f\to\infty}v(\tilde f)=0$.
Hence we are done.
\end{proof}

\begin{proof}[Proof of Lemma~\ref{lem:2003031452}]
Let $\tilde f\ge2$, and then take $\tilde r\ge2$ which solves the equation 
$$
\tilde f=(\tilde r^{1-\alpha/2}-1)/(1-\alpha/2)+1. 
$$ 
Since $(H_\alpha-\lambda)\phi=\psi$, in the distributional sense, 
we can compute by Green's formula as
\begin{equation}
\begin{split}\label{eq:2003181417}
\int_{|x|\le \tilde r}\phi\bar\psi-\bar\phi\psi \,{\rm d}x 
&= 
\tfrac12\int_{f(x)=\tilde f}(r^{-\alpha}\partial_f\phi)\bar\phi - \phi(r^{-\alpha}\partial_f\bar\phi) \,{\rm d}S_{f=\tilde f} 
\\&= 
\tfrac12\int_{f(x)=\tilde f}(\scD_f\phi)\bar\phi - \phi(\overline{\scD_f\phi}) \,{\rm d}S_{f=\tilde f} 
\\&\phantom{{}={}}{}
+i\int_{f(x)=\tilde f} r^{-\alpha/2}\bigl(\tfrac{\partial\theta}{\partial|x|}(\lambda,x)\bigr)|\phi|^2 \,{\rm d}S_{f=\tilde f}.
\end{split}
\end{equation}
Clearly, the left-hand side of \eqref{eq:2003181417} 
tends to the left-hand side of \eqref{eq:2003031457} as $\tilde f\to\infty$.
The first term on the right-hand side converges to $0$ 
by Lemma~\ref{lem:2003171406}.
In addition, noting \eqref{eq:theta-eikonal}, it holds that  
\begin{equation*}
\lim_{\tilde f\to\infty} i\int_{f(x)=\tilde f} r^{-\alpha/2}\bigl(\tfrac{\partial\theta}{\partial|x|}(\lambda,x)\bigr)|\phi|^2 \,{\rm d}S_{f=\tilde f} 
= 
i\lim_{\tilde f\to\infty}\int_{f(x)=\tilde f}|\phi|^2\,{\rm d}S_{f=\tilde f}. 
\end{equation*}
Hence by taking the limit $\tilde f\to\infty$ in the both side of \eqref{eq:2003181417} 
we have the assertion. 
\end{proof}

Let us recall the operator $\scF^\pm(\lambda, f)$ which is defined by 
\begin{equation*}
\bigl( \scF^\pm(\lambda, f)\psi\bigr)(\omega) 
= 
\tfrac1{\sqrt{2\pi}}
e^{\pm\tfrac{\pi i}4\left(\tfrac{d-\alpha/2-3}{1+\alpha/2}\right)}
r^{(d+\alpha/2-1)/2}e^{\mp i\theta(\lambda, \cdot\omega)}
\bigl(R(\lambda\pm i0)\psi\bigr)(\pm\cdot\omega), 
\end{equation*}
where $\psi\in C_0^\infty(\R^d)$ and $\omega\in \mathbb S^{d-1}$.
Then we have by Lemma~\ref{lem:2003031452} 
\begin{equation}\label{eq:2003101541}
\tfrac1{2\pi i}\langle R(\lambda+i0)\psi-R(\lambda-i0)\psi, \psi\rangle
= 
\lim_{f\to\infty}\|\scF^\pm(\lambda,f)\psi\|^2_{L^2(\mathbb S^{d-1})}.
\end{equation}
Note that 
$\|\scF^\pm(\lambda,f)\psi\|_{L^2(\mathbb S^{d-1})}$ is uniformly bounded in $f$.

\begin{proof}[Proof of Lemma~\ref{lem:2003051708}]
Let $\psi\in C^\infty_0(\R^d)$. 
It suffices to show the existence of the limit 
\begin{equation*}
\lim_{f\to\infty}\langle v, \scF(\lambda, f)\psi\rangle_{L^2(\mathbb S^{d-1})} 
\quad\text{for }\ v\in C^\infty(\mathbb S^{d-1}).
\end{equation*}

Let $v(\omega)\in C^\infty(\mathbb S^{d-1})$, and define the function $u$ by 
\begin{equation}\label{eq:2003191337}
u 
= r^{-(d+\alpha/2-1)/2}e^{i\theta(\lambda, x)}v(\omega), 
\quad \omega\in \mathbb S^{d-1}.
\end{equation}
In addition we let 
\begin{equation}\label{eq:2003191338}
g(\lambda, x)=(H_\alpha-\lambda)u.
\end{equation}
Then we have if $f>2$
\begin{equation}
\begin{split}\label{eq:2003191524}
g(\lambda, x) 
&= 
e^{i\theta(\lambda, x)}\Bigl\{ 
\big[ \tfrac{\lambda^2}{2}r^{-\alpha} 
+ q 
+ \tfrac18(d+\alpha/2-1)(d-\alpha/2+1)r^{-2} \bigr]v(\omega)
\Bigr.
\\&\phantom{{}={}{}}
\Bigl.
+ i\tfrac{\alpha\lambda}2r^{-\alpha/2-1}v(\omega)
- \tfrac12r^{-2}\bigl(\Delta_{\mathbb S^{d-1}}v\bigr)(\omega)
\Bigr\}r^{-(d+\alpha/2-1)/2}, 
\end{split}
\end{equation}
where $\Delta_{\mathbb S^{d-1}}$ is the Laplace--Beltrami operator on 
$\mathbb S^{d-1}$.
In particular, we have $g(\lambda, \cdot)\in \vB$.
We have also by straightforward calculations that for some $\delta>0$
\begin{equation}\label{eq:2003101704}
u\in\vB^*,\ \scD u\in L^2_{-1/2+\delta},\ \scD_fu=0\quad\text{if }\ f>2.
\end{equation}
Now we let $\phi=R(\lambda+i0)\psi$, 
and then by using Green's formula we have 
\begin{equation}
\begin{split}\label{eq:2003101717}
\tfrac12\int_{|x|<\tilde r}\bigl\{(\Delta u)\bar\phi - u(\Delta\bar\phi)\bigr\} \,{\rm d}x
&= 
\tfrac12\int_{f(x)=\tilde f}\bigl\{ (\scD_fu)\bar\phi - u(\overline{\scD_f\phi})\bigr\} \,{\rm d}S_{f=\tilde f}
\\&\phantom{{}={}{}}
+ i\int_{f(x)=\tilde f} r^{-\alpha/2}\tfrac{\partial\theta}{\partial|x|}(\lambda, x)u 
\bar\phi \,{\rm d}S_{f=\tilde f}.
\end{split}
\end{equation}
The left-hand side of \eqref{eq:2003101717} converges to 
$\langle u, \psi\rangle-\langle g, R(\lambda+i0)\psi\rangle$ 
as $\tilde f\to\infty$.
By noting \eqref{eq:2003101704} and 
\begin{equation*}
\int_{f(x)=\tilde f}u(\overline{\scD_\omega^*\scD_\omega\phi}) \,{\rm d}S_{f=\tilde f}
= 
\int_{f(x)=\tilde f}(\scD_\omega^*\scD_\omega u)\bar\phi \,{\rm d}S_{f=\tilde f}
\in L^1((2,\infty)),
\end{equation*}
we can see by a similar argument of the proof of Lemma~\ref{lem:2003171406} that 
\begin{equation*}
\int_{f(x)=\tilde f}\bigl\{(\scD_fu)\bar\phi-u(\overline{\scD_f\phi})\bigr\} 
\,{\rm d}S_{f=\tilde f} 
\to 0\quad \text{as }\ \tilde f\to\infty.
\end{equation*}
Moreover, we can also see that 
\begin{align*}
\lim_{\tilde f\to\infty}i\int_{f(x)=\tilde f}r^{-\alpha/2}\tfrac{\partial\theta}{\partial|x|}(\lambda, x)u\bar\phi \,{\rm d}S_{f=\tilde f} 
&= 
\lim_{\tilde f\to\infty}i\int_{f(x)=\tilde f}u\bar\phi \,{\rm d}S_{f=\tilde f} 
\\&= 
\lim_{\tilde f\to\infty}c_\alpha\langle v, \scF(\lambda, \tilde f)\psi\rangle_{L^2(\mathbb S^{d-1})}, 
\end{align*}
where 
$c_\alpha=\sqrt{2\pi}\exp{\{\tfrac{\pi i}4\bigl(\tfrac{d+\alpha/2-1}{1+\alpha/2}\bigr)\}}$. 
Therefore by taking the limit $\tilde f\to\infty$ in the both side of \eqref{eq:2003101717}, 
we have
\begin{equation}\label{eq:2003301812}
\langle u, \psi\rangle-\langle g, R(\lambda+i0)\psi\rangle
= 
\lim_{f\to\infty}c_\alpha\langle v, \scF(\lambda, f)\psi\rangle_{L^2(\mathbb S^{d-1})}. 
\end{equation}
Hence we have the assertion.
\end{proof}

By Corollary~\ref{cor:rc-real} we have for any $\beta\in[0,\beta_c)$ 
and for some $\delta>0$
\begin{equation*}
\int_2^\infty \tilde f^{2\beta-1-\delta}\left(\int_{f(x)=\tilde f}|\scD_f\phi|^2\,{\rm d}S_{f=\tilde f}\right){\rm d}\tilde f<\infty.
\end{equation*}
Thus there exists a sequence $\{f_n\}_{n\in\N}$ tending to infinity such that 
\begin{equation}\label{eq:2003191151}
f_n^{2\beta-\delta}\int_{f(x)=f_n}|\scD_f\phi|^2\,{\rm d}S_{f=f_n} \to 0 
\quad \text{as }\ n\to\infty.
\end{equation}
In particular, 
if we take $\beta=1/(1-\alpha/2)<\beta_c$ 
and $\delta=\alpha/(1-\alpha/2)$, we have
\begin{equation}\label{eq:2019031218}
\int_{f(x)=f_n}|\scD_f\phi|^2\,{\rm d}S_{f=f_n} = o(f_n^{-2}).
\end{equation}
Let us fix a sequence $\{f_n\}_{n\in\N}$ satisfying \eqref{eq:2003191151}.
In the proof of Lemma~\ref{lem:2023181624}, we use the following estimate.
\begin{proposition}\label{prop:2003201558}
Let $\psi\in C_0^\infty(\R^d)$ and $v\in C^1(\R^d)$, 
and assume 
\begin{equation*}
\sup_{f\ge2}\|v(f,\cdot)\|_{L^2(\mathbb S^{d-1})}<\infty.
\end{equation*}
Then there exists $\varepsilon>0$ such that for $2\le f_m<f_n$ 
\begin{align*}
&\quad\hspace{-3mm}
|\langle \scF(\lambda, f_m)\psi-\scF(\lambda, f_n)\psi, v \rangle_{L^2(\mathbb S^{d-1})}| 
\\&\le 
C(f_m)\left( 
\Bigl(\sup_{f\ge2}\|v(f,\cdot)\|_{L^2(\mathbb S^{d-1})}\Bigr) 
+ \left( \int_{f_m}^\infty f^{-1-\varepsilon}\|\scD_\omega v(f, \cdot)\|^2_{L^2(\mathbb S^{d-1})}\,{\rm d}f \right)^{1/2}
\right), 
\end{align*}
where $C(f_m)$ is a constant 
which is independent of $v$ and tends to $0$ as $m\to\infty$.
\end{proposition}

\begin{proof}
We show the assertion only for $v\in C^\infty(\R^d)$.
Let $v\in C^\infty(\R^d)$. 
Then we set the function $u$ and $g$ similarly to \eqref{eq:2003191337} 
and \eqref{eq:2003191338}, respectively. 
Then we have by Green's formula and \eqref{eq:2003101704}
\begin{equation}
\begin{split}\label{eq:2003191346}
&\quad\hspace{-3mm}
\int_{r_m<|x|<r_n}(u\bar\psi-g\bar\phi)\,{\rm d}x 
+ \tfrac12\int_{f(x)=f_n}u(\overline{\scD_f\phi})\,{\rm d}S_{f=f_n}
\\&\phantom{{}={}}{}
- \tfrac12\int_{f(x)=f_m}u(\overline{\scD_f\phi})\,{\rm d}S_{f=f_m}
\\&= 
i\int_{f(x)=f_n}r^{-\alpha/2}\tfrac{\partial\theta}{\partial|x|}(\lambda, x)u\bar\phi\,{\rm d}S_{f=f_n}
- i\int_{f(x)=f_m}r^{-\alpha/2}\tfrac{\partial\theta}{\partial|x|}(\lambda, x)u\bar\phi\,{\rm d}S_{f=f_m}. 
\end{split}
\end{equation}
By noting the definition $\scF(\lambda, f)$, 
we can estimate the right-hand side of \eqref{eq:2003191346} as 
\begin{equation}
\begin{split}\label{eq:2003191423}
&\quad\hspace{-3mm} 
\left| 
i\int_{f(x)=f_n}r^{-\alpha/2}\tfrac{\partial\theta}{\partial|x|}(\lambda, x)u\bar\phi\,{\rm d}S_{f=f_n}
- i\int_{f(x)=f_m}r^{-\alpha/2}\tfrac{\partial\theta}{\partial|x|}(\lambda, x)u\bar\phi\,{\rm d}S_{f=f_m} 
\right| 
\\&\ge 
\sqrt{2\pi}|\langle \scF(\lambda, f_n)\psi-\scF(\lambda, f_m)\psi, v\rangle_{L^2(\mathbb S^{d-1})}| 
\\&\phantom{{}={}}{} 
- Cr_m^{-\alpha}
\Bigl(\sup_{n\ge m}\|\scF(\lambda, f_n)\psi\|_{L^2(\mathbb S^{d-1})} \Bigr)
\Bigl(\sup_{f\ge 2}\|v(f, \cdot)\|_{L^2(\mathbb S^{d-1})} \Bigr). 
\end{split}
\end{equation}
As for the second and the third terms on the left-hand side of \eqref{eq:2003191346}, 
we can bound by the Cauchy--Schwarz inequality and \eqref{eq:2019031218} as follows. 
\begin{equation}
\begin{split}\label{eq:2003191441}
&\quad\hspace{-3mm} 
\tfrac12\left| 
\int_{f(x)=f_n}u(\overline{\scD_f\phi})\,{\rm d}S_{f=f_n}
- \int_{f(x)=f_m}u(\overline{\scD_f\phi})\,{\rm d}S_{f=f_m}
\right| 
\\&\le 
Cf_m^{-1}\Bigl(\sup_{f\ge 2}\|v(f, \cdot)\|_{L^2(\mathbb S^{d-1})} \Bigr). 
\end{split}
\end{equation}
By the Cauchy-Schwarz inequality we have for some $\varepsilon_1\in(0,1)$ 
\begin{equation}
\begin{split}
&\quad\hspace{-3mm}
\left| \int_{r_m<|x|<r_n}u\bar\psi\,{\rm d}x \right|
\\&\le 
\left(\int_{|x|>r_m}r^{-1-2\varepsilon_1}|v|^2r^{-(d-1)}\,{\rm d}x\right)^{1/2}
\left(\int_{|x|>r_m}r^{1-\alpha/2+2\varepsilon_1}|\psi|^2\,{\rm d}x\right)^{1/2}
\\&\le 
C_{\varepsilon_1} 
r_m^{-\varepsilon_1/2}\left(\sup_{f\ge 2}\|v(f, \cdot)\|_{L^2(\mathbb S^{d-1})}\right)
\!\left(\int_{|x|>r_m}r^{1-\alpha/2+2\varepsilon_1}|\psi|^2\,{\rm d}x\right)^{1/2}.
\end{split}
\end{equation}
Note that the last factor of the right-hand side is finite, 
since $\psi\in C_0^\infty(\R^d)$.
By using \eqref{eq:2003191524}, we can estimate 
\begin{equation}
\begin{split}\label{eq:2003191526}
\left| \int_{r_m<|x|<r_n}g\bar\phi\,{\rm d}x\right|
&\le 
C\int_{|x|>r_m}f^{-1-\min\{\rho, (3\alpha/2-1)/(1-\alpha/2)\}}r^{-(d+\alpha/2-1)/2}|v||\phi| 
\,{\rm d}x 
\\&\phantom{{}={}}{}
+ C\left| \int_{r_m<|x|<r_n}r^{-2}(\Delta_{\mathbb S^{d-1}}v)r^{-(d+\alpha/2-1)/2}e^{i\theta(\lambda, x)}\bar\phi\,{\rm d}x\right|. 
\end{split}
\end{equation}
Take $\varepsilon_2\in(0,\min\{\rho, (3\alpha/2-1)/(1-\alpha/2)\})$. 
By using the Cauchy--Schwarz inequality 
we can estimate the first term on the right-hand side of \eqref{eq:2003191526} as 
\begin{equation}
\begin{split}\label{eq:2003191544}
&\quad\hspace{-3mm}
C\int_{|x|>r_m}f^{-1-\min\{\rho, (3\alpha/2-1)/(1-\alpha/2)\}}r^{-(d+\alpha/2-1)/2}|v||\phi| 
\,{\rm d}x 
\\&\le 
C\Bigl(\sup_{f\ge 2}\|v(f, \cdot)\|_{L^2(\mathbb S^{d-1})} \Bigr)
\Bigl(\sup_{f\ge 2}\|\scF(\lambda, f)\psi\|_{L^2(\mathbb S^{d-1})} \Bigr)
\int_{f_m}^\infty f^{-1-\varepsilon_2}\,{\rm d}f 
\\&\le 
C_{\varepsilon_2}f_m^{-\varepsilon_2/2}
\Bigl(\sup_{f\ge 2}\|v(f, \cdot)\|_{L^2(\mathbb S^{d-1})} \Bigr)
\Bigl(\sup_{f\ge 2}\|\scF(\lambda, f)\psi\|_{L^2(\mathbb S^{d-1})} \Bigr). 
\end{split}
\end{equation}
To evaluate the second term on the right-hand side of \eqref{eq:2003191526}, 
we note that $r^{-2}\Delta_{\mathbb S^{d-1}}$ has the following expression, 
on the region $\{x\,|\,r(x)\ge2\}$,
\begin{equation*}
r^{-2}\Delta_{\mathbb S^{d-1}} 
= \sum_{j=1}^d(\partial_j-\tfrac{x_j}r\partial_{|x|})^2 
= \scD_\omega^2. 
\end{equation*}
Thus we can compute by integration by parts 
\begin{align*}
&\quad\hspace{-3mm} 
\int_{r_m<|x|<r_n}r^{-2}(\Delta_{\mathbb S^{d-1}}v)r^{-(d+\alpha/2-1)/2}e^{i\theta(\lambda, x)}\bar\phi\,{\rm d}x
\\&= 
- \int_{r_m<|x|<r_n}(\scD_\omega v)r^{-(d+\alpha/2-1)/2}e^{i\theta(\lambda, x)}
(\overline{\scD_\omega\phi})\,{\rm d}x. 
\end{align*}
We also note that since 
\begin{equation*}
|\scD_\omega\phi|^2 
= |\partial\phi-\tfrac{x}rr^{\alpha/2}\partial^f\phi|^2 
= r^\alpha\bigl(|r^{-\alpha/2}p\phi|^2-|p^f\phi|^2\bigr) \quad \text{on }\ \{x\,|\,r(x)\ge2\}, 
\end{equation*}
it holds that 
\begin{equation*}
\int_{f(x)>2} |\scD_\omega\phi|^2\,{\rm d}x 
\le \langle p_jr^\alpha\ell_{jk}p_k\rangle_\phi.
\end{equation*}
Therefore by the Cauchy--Schwarz inequality we have 
\begin{equation}
\begin{split}\label{eq:2003201520}
&\quad\hspace{-3mm} 
\left| \int_{r_m<|x|<r_n}r^{-2}(\Delta_{\mathbb S^{d-1}}v)r^{-(d+\alpha/2-1)/2}e^{i\theta(\lambda, x)}\bar\phi\,{\rm d}x\right| 
\\&\le 
Cf_m^{-\varepsilon_3/2}\left( \int_{f_m}^\infty f^{-1-\varepsilon_3}\|\scD_\omega v(f, \cdot)\|^2_{L^2(\mathbb S^{d-1})}\,{\rm d}f \right)^{1/2}
\langle p_jf^{1+2\varepsilon_3}r^\alpha\ell_{jk}p_k\rangle_\phi^{1/2}, 
\end{split}
\end{equation}
where $\varepsilon_3>0$. 
If we take $\varepsilon_3\in(0, \min\{\rho, (\alpha/2)/(1-\alpha/2)\})$, 
we can see by Corollary~\ref{cor:rc-real} 
the last factor of the right-hand side of \eqref{eq:2003201520} is finite.
Hence by \eqref{eq:2003191346}--\eqref{eq:2003201520} the assertion follows.
\end{proof}

\begin{proof}[Proof of Lemma~\ref{lem:2023181624}]
Since 
$\displaystyle\sup_{f_m\ge2}\|\scF(\lambda, f_m)\|_{L^2(\mathbb S^{d-1})}<\infty$ and 
\begin{align*}
\int_{f_m}^\infty f^{-1-\varepsilon}\|\scD_\omega \scF(\lambda, f_m)\psi\|^2_{L^2(\mathbb S^{d-1})}\,{\rm d}f
&\le 
\int_{f(x)>2} f^{-1-\varepsilon}|\scD_\omega R(\lambda+i0)\psi|^2\,{\rm d}x 
\\&\le 
\langle p_jf^{-1-\varepsilon}r^\alpha\ell_{jk}p_k\rangle_{R(\lambda+i0)\psi}<\infty, 
\end{align*}
for any $\varepsilon>0$, 
we can apply Proposition~\ref{prop:2003201558} to $v=\scF(\lambda, f_m)\psi$. 
Then we have 
\begin{align*}
&\quad\hspace{-3mm}
\left|\langle \scF(\lambda, f_m)\psi-\scF(\lambda, f_n)\psi, \scF(\lambda, f_m)\psi \rangle_{L^2(\mathbb S^{d-1})}\right| 
\\&\le 
C(f_m)\left( 
\Bigl(\sup_{f_m\ge2}\|\scF(\lambda, f_m)\psi\|_{L^2(\mathbb S^{d-1})}\Bigr) 
+ \langle p_jf^{-1-\varepsilon}r^\alpha\ell_{jk}p_k\rangle_{R(\lambda+i0)\psi}^{1/2}
\right). 
\end{align*}
Since $\scF(\lambda, f_n)\psi$ converges weakly to 
$\scF(\lambda)\psi$ in $L^2(\mathbb S^{d-1})$, 
by taking the limit $n\to\infty$ we have 
\begin{align*}
&\quad\hspace{-3mm}
\left| \|\scF(\lambda, f_m)\psi\|^2_{L^2(\mathbb S^{d-1})} 
- \langle\scF(\lambda)\psi, \scF(\lambda, f_m)\psi \rangle_{L^2(\mathbb S^{d-1})}\right| 
\\&\le 
C(f_m)\left( 
\Bigl(\sup_{f_m\ge2}\|\scF(\lambda, f_m)\psi\|_{L^2(\mathbb S^{d-1})}\Bigr) 
+ \langle p_jf^{-1-\varepsilon}r^\alpha\ell_{jk}p_k\rangle_{R(\lambda+i0)\psi}^{1/2}
\right). 
\end{align*}
Then by taking the limit $m\to\infty$ we have 
\begin{equation}\label{eq:2003201625}
\lim_{m\to\infty} \|\scF(\lambda, f_m)\psi\|^2_{L^2(\mathbb S^{d-1})}
= \|\scF(\lambda)\psi\|^2_{L^2(\mathbb S^{d-1})}.
\end{equation}
Hence by \eqref{eq:2003201625} and \eqref{eq:2003051710} we can see that 
\begin{align*}
\|\scF(\lambda, f_m)\psi-\scF(\lambda)\psi\|^2_{L^2(\mathbb S^{d-1})} 
&= 
\|\scF(\lambda, f_m)\psi\|^2_{L^2(\mathbb S^{d-1})} 
+ \|\scF(\lambda)\psi\|^2_{L^2(\mathbb S^{d-1})} 
\\&\phantom{{}={}}{}
- 2\mathop{\mathrm{Re}}\langle\scF(\lambda, f_m)\psi, \scF(\lambda)\psi \rangle_{L^2(\mathbb S^{d-1})} 
\\&\to 0 \quad \text{as }\ m\to\infty.
\end{align*}
This implies the assertion.
\end{proof}

Now let us prove Theorem~\ref{thm:existence-of-wave-matrix}.
\begin{proof}[Proof of Theorem~\ref{thm:existence-of-wave-matrix}]
Since $\scF(\lambda, f)\psi$ converges weakly to 
$\scF(\lambda)\psi$ in $L^2(\mathbb S^{d-1})$, we have 
\begin{equation}\label{eq:2003211339}
\|\scF(\lambda)\psi\|_{L^2(\mathbb S^{d-1})} 
\le 
\liminf_{f\to\infty}\|\scF(\lambda, f)\psi\|_{L^2(\mathbb S^{d-1})}. 
\end{equation}
On the other hand, by using the sequence $\{f_n\}_{n\in\N}$ 
appearing in Lemma~\ref{lem:2023181624} we can see that 
\begin{equation}\label{eq:2003211340}
\liminf_{f\to\infty}\|\scF(\lambda, f)\psi\|_{L^2(\mathbb S^{d-1})}
\le 
\lim_{n\to\infty}\|\scF(\lambda, f_n)\psi\|_{L^2(\mathbb S^{d-1})} 
= 
\|\scF(\lambda)\psi\|_{L^2(\mathbb S^{d-1})}. 
\end{equation}
Thus by \eqref{eq:2003211339}, \eqref{eq:2003211340} and \eqref{eq:2003101541} 
we have 
\begin{equation*}
\|\scF(\lambda)\psi\|_{L^2(\mathbb S^{d-1})} 
= 
\liminf_{f\to\infty}\|\scF(\lambda, f)\psi\|_{L^2(\mathbb S^{d-1})} 
= 
\lim_{f\to\infty}\|\scF(\lambda, f)\psi\|_{L^2(\mathbb S^{d-1})}. 
\end{equation*}
By using this equality we can conclude \eqref{eq:2003161503}. 
The equation \eqref{eq:2003161512} follows immediately 
from \eqref{eq:2003161503} and \eqref{eq:2003101541}.
\end{proof}


\subsection{Stationary wave operators}
\label{sec:Stationary wave operators}

First we give expressions formulae of $\scF^\pm(\lambda)$.
\begin{proposition}\label{prop:2005302142}
For any $\psi\in\vB$ the vectors $\scF^\pm(\lambda)\psi$ have expressions: 
\begin{equation}\label{eq:2005302149}
\scF^\pm(\lambda)\psi 
= \lim_{R\to\infty}R^{-1}\int_R^{2R}\scF^\pm(\lambda, f)\psi\,{\rm d}f, 
\end{equation}
respectively, in $L^2(\mathbb S^{d-1})$.
Moreover $R^{-1}\int_R^{2R}\scF^\pm(\lambda, f)\psi\,{\rm d}f$ converge locally uniformly in $\lambda\in\R$.
\end{proposition}

\begin{proof}
First we prove \eqref{eq:2005302149} for $\psi\in\vB$.
We note that \eqref{eq:2005302149} for $\psi\in C_0^\infty(\R^d)$ 
obviously holds by Theorem~\ref{thm:existence-of-wave-matrix}. 
Moreover the left-hand side is continuous for $\psi\in\vB$. 
Thus it suffices to verify the existence and continuity of 
the right-hand side of \eqref{eq:2005302149} for $\psi\in\vB$.
To do this we show the following estimate by similar way to \cite{IS2}: 
For any $\psi\in\vB$ 
\begin{equation}\label{eq:2006010039}
\sup_{R>2}\left\|\dfrac1R\int_R^{2R}\scF^\pm(\lambda, f)\psi\,{\rm d}f\right\|_{L^2(\mathbb S^{d-1})}
\le C\|R(\lambda\pm i0)\psi\|_{\vB^*}.
\end{equation}

For any $R>2$ we choose $n\in\N$ satisfying $2^n\le 2R<2^{n+1}$. 
Then by using the Cauchy--Schwarz inequality we have 
\begin{align*}
\left\|\dfrac1R\int_R^{2R}\scF^\pm(\lambda, f)\psi\,{\rm d}f\right\|^2_{L^2(\mathbb S^{d-1})} 
&\le 
\dfrac1R\int_R^{2R}\left\|\scF^\pm(\lambda, f)\psi\right\|^2_{L^2(\mathbb S^{d-1})}\,{\rm d}f 
\\&\le 
\dfrac1{2\pi R}\int_{2^{n-1}}^{2^{n+1}}r^{d+\alpha/2-1}\|R(\lambda\pm i0)\psi\|^2_{L^2(\mathbb S^{d-1})}\,{\rm d}f 
\\&\le 
\dfrac3{2\pi}\|R(\lambda\pm i0)\psi\|_{\vB^*}^2.
\end{align*}
By taking supremum in $R>2$ we have \eqref{eq:2006010039} 
and then, we conclude that \eqref{eq:2005302149} holds for any $\psi\in\vB$. 

Next we prove the latter assertion.
We note that we may assume that $\psi\in C_0^\infty(\R^d)$. 
By the Cauchy--Schwarz inequality we have 
\begin{equation}
\begin{split}\label{eq:2006022012}
&\quad\hspace{-3mm}
\left\|\scF^\pm(\lambda)\psi-\dfrac1{R_1}\int_{R_1}^{2R_1}\scF^\pm(\lambda, f_1)\psi
\,{\rm d}f_1\right\|^2_{L^2(\mathbb S^{d-1})} 
\\&\le 
\lim_{R_2\to\infty}\dfrac1{R_1}\int_{R_1}^{2R_1}\dfrac1{R_2}\int_{R_2}^{2R_2}
\|\scF^\pm(\lambda, f_2)\psi-\scF^\pm(\lambda, f_1)\psi\|^2_{L^2(\mathbb S^{d-1})}
\,{\rm d}f_2\,{\rm d}f_1.
\end{split}
\end{equation}
We write 
\begin{align*}
&\quad\hspace{-3mm}
\|\scF^\pm(\lambda, f_2)\psi-\scF^\pm(\lambda, f_1)\psi\|^2_{L^2(\mathbb S^{d-1})} 
\\&= 
\|\scF^\pm(\lambda, f_2)\psi\|^2_{L^2(\mathbb S^{d-1})} 
- \|\scF^\pm(\lambda, f_1)\psi\|^2_{L^2(\mathbb S^{d-1})} 
\\&\phantom{{}={}}{}
+ 2\mathop{\mathrm{Re}}\langle\scF^\pm(\lambda, f_1)\psi-\scF^\pm(\lambda, f_2)\psi, \scF^\pm(\lambda, f_1)\psi \rangle_{L^2(\mathbb S^{d-1})}
\end{align*}
and substitute this into the right-hand side of \eqref{eq:2006022012}, 
then we can obtain 
\begin{align*}
&\quad\hspace{-3mm}
\left\|\scF^\pm(\lambda)\psi-\dfrac1{R_1}\int_{R_1}^{2R_1}\scF^\pm(\lambda, f_1)\psi
\,{\rm d}f_1\right\|^2_{L^2(\mathbb S^{d-1})} 
\\&\le 
\|\scF^\pm(\lambda)\psi\|^2_{L^2(\mathbb S^{d-1})} 
- \dfrac1{R_1}\int_{R_1}^{2R_1}\|\scF^\pm(\lambda, f_1)\psi\|^2_{L^2(\mathbb S^{d-1})}
\,{\rm d}f_1 
\\&\phantom{{}={}}{}
+ \dfrac2{R_1}\int_{R_1}^{2R_1}\mathop{\mathrm{Re}}
\langle\scF^\pm(\lambda, f_1)\psi-\scF^\pm(\lambda)\psi, \scF^\pm(\lambda, f_1)\psi \rangle_{L^2(\mathbb S^{d-1})}
\,{\rm d}f_1. 
\end{align*}
By the proofs of Lemmas~\ref{lem:2003031452} and~\ref{lem:2003171406}, 
the second term on the right-hand side tends to 
$-\|\scF^\pm(\lambda)\psi\|^2_{L^2(\mathbb S^{d-1})}$ 
locally uniformly in $\lambda\in\R$. 
Similarly, by the proof of Lemma~\ref{lem:2023181624}, 
the third term tends to $0$ locally uniformly in $\lambda\in\R$. 
Hence we have the assertion.
\end{proof}


\begin{proof}[Proof of Proposition~\ref{prop:2005302258}]
By \eqref{eq:2003161512} and Stone's formula we can compute as, 
for any $\psi\in\vB$, 
\begin{align*}
\|\vF^\pm\psi\|_{\widetilde\vH}^2 
&= 
\int_\R \|\scF^\pm(\lambda)\psi\|_{L^2(\mathbb S^{d-1})}^2\,{\rm d}\lambda 
= 
\int_\R \dfrac{{\rm d}\langle E(\lambda)\psi, \psi\rangle}{{\rm d}\lambda}
\,{\rm d}\lambda 
= 
\|\psi\|_{\vH}^2, 
\end{align*}
where $E(\cdot)$ is the spectral projection of $H_\alpha$. 
Thus $\vF^\pm$ extend to isometries $\vH\to\widetilde\vH$.
To verify $\vF^\pm H_\alpha\subseteq M_\lambda\vF^\pm$ 
it suffices to show that 
\begin{equation}\label{eq:2006232356}
\vF^\pm(H_\alpha-i)^{-1}\psi=(M_\lambda-i)^{-1}\vF^\pm\psi 
\quad\text{for any }\ \psi\in\vB. 
\end{equation}
By using the expressions \eqref{eq:2005302149}, the resolvent equations 
\begin{equation*}
R(\lambda\pm i0)R(i) = (\lambda-i)^{-1}R(\lambda\pm i0)-(\lambda-i)^{-1}R(i) 
\end{equation*}
and the Cauchy--Schwarz inequality, we obtain 
\begin{equation*}
\scF^\pm(\lambda)R(i)\psi 
= \lim_{R\to\infty}\dfrac1R\int_R^{2R}\scF^\pm(\lambda, f)R(i)\psi\,{\rm d}f 
= (\lambda-i)^{-1}\scF^\pm(\lambda)\psi.
\end{equation*}
Thus we have \eqref{eq:2006232356}.
\end{proof}


\section{Wave matrix and scattering matrix}
\label{sec:Wave matrix and scattering matrix}

In this section we prove Proposition~\ref{prop:2004031943} and 
Theorem~\ref{thm:characterize-B^*-ef}. 
\begin{proof}[Proof of Proposition~\ref{prop:2004031943}]
We consider only for the upper sign. 
It immediately follows from \eqref{eq:2003301812} that for any $\psi\in\vB$
\begin{equation*}
\langle \phi_\lambda^+[v], \psi \rangle 
- \langle \psi_\lambda^+[v], R(\lambda+i0)\psi \rangle 
= 
\langle v, \scF^+(\lambda)\psi \rangle_{L^2(\mathbb S^{d-1})}. 
\end{equation*}
This implies \eqref{eq:2004032011}.
\end{proof}

To show Theorem~\ref{thm:characterize-B^*-ef}, 
we introduce and prove several lemmas.
First we prove uniqueness. 
\begin{lemma}\label{lem:2008082135}
Suppose $\xi_\pm\in L^2(\mathbb S^{d-1})$ and $\phi\in \mathcal E_\lambda$ satisfy 
\begin{equation}\label{eq:2007071457}
\phi-\phi_\lambda^+[\xi_+]-\phi_\lambda^-[\xi_-]\in \vB_0^*.
\end{equation}
Then we have 
\begin{align}\label{eq:2005241344}
\|\xi_+\|^2+\|\xi_-\|^2 &= \lim_{n\to\infty}2^{-n}\int_{2^n<f<2^{n+1}}2\pi|\phi|^2{\rm d}x, \\
\|\xi_+\| &= \|\xi_-\|. \label{eq:2005241345}
\end{align}
In particular, the two quantities in $\{\xi_+, \xi_-, \phi\}$ satisfying \eqref{eq:2007071457}
uniquely determined by the other one of them.
\end{lemma}
\begin{proof}
First we prove \eqref{eq:2005241344}.
We can compute the right-hand side of \eqref{eq:2005241344} as
\begin{align*}
&\quad\hspace{-3mm}
\lim_{n\to\infty}2^{-n}\int_{2^n<f<2^{n+1}}2\pi|\phi|^2\,{\rm d}x 
\\&= 
\lim_{n\to\infty}2^{-n}\int_{2^n<f<2^{n+1}}2\pi|\phi_\lambda^+[\xi_+]+\phi_\lambda^-[\xi_-]|^2\,{\rm d}x 
\\&= 
\|\xi_+\|^2 + \|\xi_-\|^2 
+ 2\mathop{\mathrm{Re}}\lim_{n\to\infty}2^{-n}\int_{2^n}^{2^{n+1}}
e^{-\frac{\pi i}2\bigl(\frac{d+\alpha/2-1}{1+\alpha/2}\bigr)}e^{2i\theta}\,{\rm d}f 
\int_{\mathbb S^{d-1}}\xi_+(\omega)\overline{\xi_-(-\omega)}\,{\rm d}\omega. 
\end{align*}
Since $e^{2i\theta}=(2i(r^\alpha+\lambda))^{-1}\tfrac{d}{df}e^{2i\theta}$, 
by integrating by parts we can see that the last term vanishes.
Hence we have \eqref{eq:2005241344}. 

Next we prove \eqref{eq:2005241345}.
Noting that $AR(i)\in\vL(\vB^*)$ and $A\phi=(\lambda-i)AR(i)\phi\in\vB^*$, 
we obtain
\begin{align*}
0 &= 
\lim_{n\to\infty}\langle i[H_\alpha, \chi_n]\rangle_\phi 
\\&= 
\lim_{n\to\infty}\langle A\chi_n'\rangle_\phi 
\\&= 
\lim_{n\to\infty}\langle A\phi, 
\chi_n'\bigl(\phi_\lambda^+[\xi_+]+\phi_\lambda^-[\xi_-]\bigr)\rangle 
\\&= 
\lim_{n\to\infty}\langle \phi, 
\chi_n'\bigl(A\phi_\lambda^+[\xi_+]+A\phi_\lambda^-[\xi_-]\bigr)\rangle 
\\&= 
\lim_{n\to\infty}\langle \phi, 
\chi_n'\bigl(\phi_\lambda^+[\xi_+]-\phi_\lambda^-[\xi_-]\bigr)\rangle 
\\&= 
\tfrac1{2\pi}\left(\|\xi_+\|^2 - \|\xi_-\|^2\right).
\end{align*}
This implies \eqref{eq:2005241345}.

Finally, the uniqueness statement follows from \eqref{eq:2005241344}, 
\eqref{eq:2005241345}, linearity of $\phi_\lambda^\pm$ and 
Rellich's theorem, or the absence of $\vB_0^*$-eigenfunctions for $H_\alpha$.
\end{proof}

Next we construct $\phi\in\vE_\lambda$ from $\xi_\pm\in C^\infty(\mathbb S^{d-1})$. 
\begin{lemma}\label{lem:2007111232}
Let $\lambda\in\R$. 
For any $\xi_-\in C^\infty(\mathbb S^{d-1})$ we define 
$\phi\in \vE_\lambda$ and $\xi_+\in L^2(\mathbb S^{d-1})$ by 
\begin{align*}
\phi = \phi_\lambda^-[\xi_-]-R(\lambda+i0)\psi_\lambda^-[\xi_-], 
\quad 
\xi_+ = -2\pi i \scF^+(\lambda)\psi^-_\lambda[\xi_-].
\end{align*}
Then \eqref{eq:2005232312} and \eqref{eq:2005241246} 
hold for \{$\xi_+, \xi_-, \phi$\}.
\end{lemma}

\begin{proof}
By \eqref{eq:2004032011} and \eqref{eq:2004021457}, 
we can see \eqref{eq:2005241246} holds.
Moreover we can obtain \eqref{eq:2005232312} by the following calculation.
\begin{align*}
&\quad\hspace{-3mm} 
\phi-\phi_\lambda^+[\xi_+]-\phi_\lambda^-[\xi_-] 
\\&= 
2\pi i \phi_\lambda^+\left[\scF^+(\lambda)\psi^-_\lambda[\xi_-]\right] 
- R(\lambda+i0)\psi_\lambda^-[\xi_-] 
\\&= 
\sqrt{2\pi}i e^{-\frac{\pi i}4\left(\frac{d+\alpha/2-1}{1+\alpha/2}\right)}
r^{-(d+\alpha-1)/2}e^{i\theta}\scF^+(\lambda)\psi^-_\lambda[\xi_-] 
- R(\lambda+i0)\psi_\lambda^-[\xi_-] 
\\&= 
\sqrt{2\pi}e^{-\frac{\pi i}4\left(\frac{d-\alpha/2-3}{1+\alpha/2}\right)}
r^{-(d+\alpha/2-1)/2}e^{i\theta}
\left( \scF^+(\lambda)\psi^-_\lambda[\xi_-]- \scF^+(\lambda, f)\psi_\lambda^-[\xi_-]
\right).
\end{align*}
\end{proof}

We can obtain a similar result for $\xi_+\in C^\infty(\mathbb S^{d-1})$ given first.
\begin{lemma}\label{lem:2007181542}
Let $\lambda\in\R$. 
For any $\xi_+\in C^\infty(\mathbb S^{d-1})$ we define 
$\phi\in \vE_\lambda$ and $\xi_-\in L^2(\mathbb S^{d-1})$ by 
\begin{align*}
\phi = \phi_\lambda^+[\xi_+]-R(\lambda-i0)\psi_\lambda^+[\xi_+], 
\quad 
\xi_- = 2\pi i \scF^-(\lambda)\psi^+_\lambda[\xi_+].
\end{align*}
Then \eqref{eq:2005232312} and \eqref{eq:2005241246} 
hold for \{$\xi_+, \xi_-, \phi$\}.
\end{lemma}

Let us construct $\xi_\pm\in L^2(\mathbb S^{d-1})$ from $\phi\in\vE_\lambda$. 
\begin{lemma}\label{lem:2007211540}
Let $\lambda\in\R$. 
For any $\phi\in\vE_\lambda$ there exist 
$\xi_\pm\in L^2(\mathbb S^{d-1})$ 
such that \eqref{eq:2005241246} holds. 
\end{lemma}

\begin{proof}
We use similar scheme to \cite{IS2, Sk}.
By the definition of $S(\lambda)$, 
to verify \eqref{eq:2005241246} 
it suffices to show that there exists $\xi\in L^2(\mathbb S^{d-1})$ such that 
$\phi=\scF^+(\lambda)^*\xi$ holds.

We take and fix a function $\eta\in C_0^\infty(\R)$ which satisfies $\eta(t)=t$ 
in neighborhood of $t=\lambda$. 
We define $\phi_\pm\in\vB^*$ and 
$\xi_n\in L^2(\mathbb S^{d-1}), n\in\N$ for fixed large $m\in\N$ by 
\begin{equation*}
\phi_\pm=\dfrac1{2a_0}\bar\chi_m(A\pm a_0)\phi,
\quad 
\xi_n=2\pi i\scF^+(\lambda)\chi_n\left(\eta(H)-\lambda\right)\phi_+,
\end{equation*}
where $a_0=r^{-\alpha/2}\sqrt{2\lambda-2q_0+r^\alpha}$.
First we see $\xi_n$ is bounded uniformly in $n\in\N$. 
By commuting $\chi_n$ and $\eta(H)$ 
and noting $\scF^+(\lambda)(\eta(H)-\lambda)=0$, 
we have for $v\in C^\infty(\mathbb S^{d-1}), \|v\|_{L^2(\mathbb S^{d-1})}=1$, 
\begin{align*}
\langle v, \xi_n\rangle_{L^2(\mathbb S^{d-1})} 
= -2\pi i\langle \scF^+(\lambda)^*v, [\chi_n, \eta(H)]\phi_+\rangle_{\vB^*\times\vB}. 
\end{align*}
By the Helffer--Sj\"ostrand formula (cf. \cite{HS})
\begin{equation*}
\eta(H)=\int_\C R(z)\,{\rm d}\mu(z);\quad 
\,{\rm d}\mu(z)=-(2\pi i)^{-1}\bar\partial_{z}\tilde\eta(z)\,{\rm d}z\,{\rm d}\bar z,
\end{equation*}
where $\bar\partial_z=\tfrac12(\partial_x+i\partial_y)$ for $z=x+iy$ and 
$\tilde\eta$ is a almost analytic extension of $\eta$, 
we have the expression 
\begin{equation}\label{eq:2008071833}
[\chi_n, \eta(H)] 
= -i\int_\C R(z)(A\chi_n'+\tfrac{i}2|\partial f|^2\chi_n'')R(z)\,{\rm d}\mu(z).
\end{equation}
Then we can write 
\begin{align*}
[\chi_n, \eta(H)]\phi_+ 
&= 
i(A\chi_n'+i|\partial f|^2\chi_n'')\eta'(H)\phi_+ 
+\tfrac12\int_\C R(z)|\partial f|^2\chi_n''R(z)\phi_+\,{\rm d}\mu(z)
\\&\phantom{{}={}}{}
+ \int_\C R(z)\left([H, iA]_{\chi_n'}+z|\partial f|^2\chi_n''\right)R(z)^2\phi_+\,{\rm d}\mu(z).
\end{align*}
Since $[H, iA]_{\chi_n'}$ can be expressed as, cf. \cite[Lemma~2.2]{Ita2}, 
\begin{equation*}
[H, iA]_{\chi_n'} 
= \mathop{\mathrm{Re}}\left(\gamma_1H\right)+A\gamma_2A+\gamma_3, 
\end{equation*}
where $\gamma_j, j=1,2,3$ are certain real-valued functions which satisfying 
$\mathop{\mathrm{supp}}\gamma_j\subseteq\mathop{\mathrm{supp}}\chi_n'$ 
and $|\gamma_j|=\vO(f^{-2})$, 
we can see the second and the third terms are bounded in $\vB$ uniformly in $n\in\N$. 
Thus we can estimate as 
\begin{align*}
|\langle v, \xi_n\rangle_{L^2(\mathbb S^{d-1})}| 
\le 
C_1\left(\|A\scF^+(\lambda)^*v\|_{\vB^*}+\|\scF^+(\lambda)^*v\|_{\vB^*}\right)\le C_2.
\end{align*}
Therefore the sequence $\{\xi_n\}_{n\in\N}\subset L^2(\mathbb S^{d-1})$ is bounded. 
Let us choose a weakly convergent subsequence of $\{\xi_n\}_{n\in\N}$ and 
denote its weak limit by $\xi$. 
By changing notation, we may assume 
$\mathop{\mathrm{w\text{-}lim}}_{n\to\infty}\xi_n=\xi\in L^2(\mathbb S^{d-1})$. 
For this $\xi$, we show $\scF^+(\lambda)^*\xi=\phi$. 
We introduce the function $\check\eta(t):=(\eta(t)-\lambda)(t-\lambda)^{-1}$ 
and compute 
\begin{align*}
&\quad\hspace{-3mm}\scF^+(\lambda)^*\xi
\\&= 
\mathop{\mathrm{w^*\text{-}\vB^*\text{-}lim}}_{n\to\infty}
2\pi i \scF^+(\lambda)^*\scF^-(\lambda)\chi_n(\eta(H)-\lambda)\phi_+
\\&= 
\mathop{\mathrm{w^*\text{-}\vB^*\text{-}lim}}_{n\to\infty}
\left(R(\lambda+i0)-R(\lambda-i0)\right)\chi_n(\eta(H)-\lambda)\phi_+
\\&= 
\mathop{\mathrm{w^*\text{-}\vB^*\text{-}lim}}_{n\to\infty}
\bigl(\check\eta(H)\chi_n\phi 
+R(\lambda+i0)[\chi_n, \eta(H)]\phi_+
+R(\lambda-i0)[\chi_n, \eta(H)](\phi-\phi_+)\bigr)
\\&= 
\phi + \mathop{\mathrm{w^*\text{-}\vB^*\text{-}lim}}_{n\to\infty}
\bigl(R(\lambda+i0)[\chi_n, \eta(H)]\phi_+
+R(\lambda-i0)[\chi_n, \eta(H)](\phi-\phi_+)\bigr).
\end{align*}
By Corollary~\ref{cor:uniqueness-result-refine} and \eqref{eq:2008071833}, we have 
\begin{align*}
&\quad\hspace{-3mm}
\mathop{\mathrm{w^*\text{-}\vB^*\text{-}lim}}_{n\to\infty}
R(\lambda+i0)[\chi_n, \eta(H)]\phi_+
\\&= 
-i\mathop{\mathrm{w^*\text{-}\vB^*\text{-}lim}}_{n\to\infty}
\int_\C R(z)R(\lambda+i0)A\chi_n'R(z)\,{\rm d}\mu(z)\phi_+
\\&= 
i\mathop{\mathrm{w^*\text{-}\vB^*\text{-}lim}}_{n\to\infty}
\int_\C R(z)R(\lambda+i0)a_0\chi_n'R(z)\,{\rm d}\mu(z)\phi_+
\\&= 
-\tfrac{i}2\mathop{\mathrm{w^*\text{-}\vB^*\text{-}lim}}_{n\to\infty}
\eta'(H)R(\lambda+i0)\chi_n'(A+a_0)\phi
\\&= 
-\tfrac{i}2\mathop{\mathrm{w^*\text{-}\vB^*\text{-}lim}}_{n\to\infty}
\eta'(H)R(\lambda+i0)(A+a_0)\chi_n'\phi
=0.
\end{align*}
By noting $\phi-\phi_+=\chi_m\phi-\phi_-$, we can see by similar argument that 
\begin{equation*}
\mathop{\mathrm{w^*\text{-}\vB^*\text{-}lim}}_{n\to\infty}
R(\lambda-i0)[\chi_n, \eta(H)](\phi-\phi_+)=0.
\end{equation*}
Hence we are done.
\end{proof}

\begin{proof}[Proof of Theorem~\ref{thm:characterize-B^*-ef}]
First we let $\xi_-\in L^2(\mathbb S^{d-1})$ and choose 
a sequence $\{\xi_{-, n}\}\subset C^\infty(\mathbb S^{d-1})$ such that 
$\xi_{-, n}$ converges to $\xi_-$ in $L^2(\mathbb S^{d-1})$ as $n\to\infty$. 
Then by Lemma~\ref{lem:2007111232} we have 
\begin{equation*}
\scF^-(\lambda)^*\xi_{-, n}
-\phi_\lambda^+[S(\lambda)\xi_{-, n}]
-\phi_\lambda^-[\xi_{-, n}]\in \vB_0^*.
\end{equation*}
Since $\scF^-(\lambda)^*, \phi_\lambda^\pm[\,\cdot\,]$ and $S(\lambda)$ 
are continuous, we can obtain by taking the limit $n\to\infty$ 
\begin{equation*}
\scF^-(\lambda)^*\xi_-
-\phi_\lambda^+[S(\lambda)\xi_-]
-\phi_\lambda^-[\xi_-]\in \vB_0^*.
\end{equation*}
Thus it is shown that \eqref{eq:2005232312} and \eqref{eq:2005241246} 
hold when $\xi_-$ is given first. 
By Lemma~\ref{lem:2007181542} and similar argument, 
we can see that these hold when $\xi_+$ is given first. 
By Lemmas~\ref{lem:2008082135} and~\ref{lem:2007211540} 
we can conclude (i). 

To verify (ii), let us show \eqref{eq:2005241247}. 
We have by straightforward calculations 
\begin{equation*}
(A\mp a_\pm)\phi_\lambda^\pm[\xi]\in\vB_0^*
\end{equation*}
for any $\xi\in C^\infty(\mathbb S^{d-1})$, 
and then we can deduce from \eqref{eq:2004032011} that 
\begin{equation*}
(A\pm a_\pm)\scF^\pm(\lambda)^*\xi \mp 2a_\pm\phi_\lambda^\pm[\xi] 
= 
(A\mp a_\pm)\phi_\lambda^\pm[\xi] 
-(A\pm a_\pm)R(\lambda \mp i0)\psi_\lambda^\pm[\xi]
\in\vB_0^*.
\end{equation*}
Therefore by the definition of $\phi_\lambda^\pm[\,\cdot\,]$ we obtain, 
for any $\xi\in C^\infty(\mathbb S^{d-1})$,
\begin{equation}
\begin{split}\label{eq:2008082208}
\xi 
&= \pm\frac12c_\pm\lim_{R\to\infty}\int_R^{2R}r^{(d+\alpha/2-1)/2}e^{\mp i\theta}
\left[\frac1{a_\pm}(A\pm a_\pm)\scF^\pm(\lambda)^*\xi\right]\!(f, \pm\cdot)\,{\rm d}f 
\\&= 
\pm\frac12c_\pm\lim_{R\to\infty}\int_R^{2R}r^{(d+\alpha/2-1)/2}e^{\mp i\theta}
\left[(A\pm a_0)\scF^\pm(\lambda)^*\xi\right]\!(f, \pm\cdot)\,{\rm d}f, 
\end{split}
\end{equation}
where $c_\pm=\sqrt{2\pi}\exp\{\pm\tfrac{\pi i}4(\tfrac{d+\alpha/2-1}{1+\alpha/2})\}$. 
Since 
\begin{equation*}
(A\pm a_\pm)\scF^\pm(\lambda)^*
=(\lambda-i)\{(A\pm a_\pm)R(i)\}\scF^\pm(\lambda)^*
\in \vL\bigl(L^2(\mathbb S^{d-1}), \vB^*\bigr),
\end{equation*}
we can see that \eqref{eq:2008082208} holds for all $\xi\in L^2(\mathbb S^{d-1})$. 
Hence we have shown (ii). 

By \eqref{eq:2005241246} and Lemma~\ref{lem:2008082135} we have 
\begin{equation*}
\|\xi_\pm\|_{L^2(\mathbb S^{d-1})}^2 
\le \pi\|\scF^\pm(\lambda)^*\xi_\pm\|_{\vB^*}^2.
\end{equation*}
In particular, $\mathop{\mathrm{Ker}}\scF^\pm(\lambda)^*=\{0\}$ and 
$\mathop{\mathrm{Ran}}\scF^\pm(\lambda)^* (=\vE_\lambda)$ 
are closed in $\vB^*$. 
Then by closed range theorem (cf.~e.g.~\cite{Y}), 
the range of $\scF^\pm(\lambda): \vB\to L^2(\mathbb S^{d-1})$ 
coincides with $L^2(\mathbb S^{d-1})$, respectively. 
Hence we are done. 
\end{proof}


\appendix
\section{Proof of Theorem~\ref{thm:rc-complex}}
\label{Appen:proof-of-thm}

First we note that the following lemma holds.
\begin{lemma}\label{lem:1912161547}
Let $z\in I_\pm$.
\begin{itemize}
\item[(1)] 
There exists $C>0$ such that for any $z\in I_\pm$ 
and $x\in\mathop{\mathrm{supp}\bar\chi_{m+1}}$
\begin{equation}\label{eq:1912171445}
|a|\le C, \quad \pm\mathop{\mathrm{Im}}a \ge \tfrac\alpha2r^{-\alpha/2-1}-Cr^{-3\alpha/2-1},
\end{equation}
respectively.
\item[(2)]
One can rewrite $H_\alpha-z$ on $\mathop{\mathrm{supp}}\bar\chi_{m+1}$ as
\begin{equation}\label{eq:1912161558}
H_\alpha-z 
= \tfrac12(A\pm a)r^\alpha(A\mp a) 
+ \tfrac12p_jr^\alpha\ell_{jk}p_k + q_2, 
\end{equation} 
where $q_2$ is a certain complex-valued function which satisfies 
\begin{equation}\label{eq:1912171524}
\bar\chi_{m+1}q_2=\vO(r^{-1}f^{-1-\min\{\rho,(3\alpha/2)/(1-\alpha/2)\}}).
\end{equation}
\end{itemize}
\end{lemma}

\begin{proof}
The bounds in \eqref{eq:1912171445} clearly hold 
by the definition of \eqref{eq:improved-phase}.
Thus we discuss on (2). 
By noting the expressions \eqref{eq:defA} and \eqref{eq:f-ell}, 
we can compute, on $\mathop{\mathrm{supp}}\bar\chi_{m+1}$, as
\begin{align*}
H_\alpha-z 
&= \tfrac12(p^f)^*r^\alpha p^f 
+ \tfrac12p_jr^\alpha\ell_{jk}p_k 
- \tfrac12r^\alpha + q - z 
\\&= 
\tfrac12Ar^\alpha A 
+ \tfrac12p_jr^\alpha\ell_{jk}p_k 
- \tfrac12r^\alpha + q 
+ \tfrac18r^\alpha(\Delta f)^2 
+ \tfrac{\alpha}4r^{\alpha/2-1}(\Delta f) 
\\&\phantom{{}={}}{}
+ \tfrac14r^\alpha(\partial^f\Delta f) 
- z 
\\&= 
\tfrac12(A\pm a)r^\alpha(A\mp a) 
+ \tfrac12p_jr^\alpha\ell_{jk}p_k 
\pm \tfrac12(p^fr^\alpha a) 
+ \tfrac12r^\alpha a^2
\\&\phantom{{}={}}{}
- \tfrac12r^\alpha + q_0 
+ \tfrac\alpha4r^{-2}
- z 
\end{align*}
Therefore by letting 
\begin{equation*}
q_2 
= 
\pm \tfrac12(p^fr^\alpha a) 
+ \tfrac12r^\alpha a^2
- \tfrac12r^\alpha + q_0 
+ \tfrac\alpha4r^{-2}
- z,
\end{equation*}
we have the expression \eqref{eq:1912161558}.
In addition, we can compute $q_2$ on $\mathop{\mathrm{supp}}\bar\chi_{m+1}$ as
\begin{align*}
q_2 
&= 
\mp \tfrac{i}2\partial^f\left\{ r^{\alpha/2}\sqrt{2z-2q_0+r^\alpha} \pm \tfrac{i\alpha}2r^{\alpha/2-1} \mp \tfrac{i\alpha}2\tfrac{z-q_0}{2z-2q_0+r^\alpha}r^{\alpha/2-1} \right\}
\\&\phantom{{}={}}{}
+ \tfrac12\left\{ \sqrt{2z-2q_0+r^\alpha} \pm \tfrac{i\alpha}2r^{-1} \mp \tfrac{i\alpha}2\tfrac{z-q_0}{2z-2q_0+r^\alpha}r^{-1} \right\}^2
- \tfrac12r^\alpha + q_0 
- z 
+ \tfrac\alpha4r^{-2}
\\&= 
\pm \tfrac{i\alpha}4r^{-1}\sqrt{2z-2q_0+r^\alpha} 
\mp \tfrac{i}4\bigl(\alpha r^{\alpha-1}-2(\partial^rq_0) \bigr)/\sqrt{2z-2q_0+r^\alpha} 
\\&\phantom{{}={}}{}
\mp \tfrac{i\alpha}2r^{-1}(z-q_0)/\sqrt{2z-2q_0+r^\alpha} 
- \tfrac\alpha4\bigl(\partial^r \tfrac{z-q_0}{2z-2q_0+r^\alpha}\bigr)r^{-1}
\\&\phantom{{}={}}{}
+ \tfrac\alpha4\tfrac{z-q_0}{2z-2q_0+r^\alpha}r^{-2} 
- \tfrac{\alpha^2}8r^{-2}\bigl(\tfrac{z-q_0}{2z-2q_0+r^\alpha}\bigr)^2
\\&= 
\pm \tfrac{i}2(\partial^rq_0)/\sqrt{2z-2q_0+r^\alpha} 
- \tfrac\alpha4\bigl(\partial^r \tfrac{z-q_0}{2z-2q_0+r^\alpha}\bigr)r^{-1}
\\&\phantom{{}={}}{}
+ \tfrac\alpha4\tfrac{z-q_0}{2z-2q_0+r^\alpha}r^{-2} 
- \tfrac{\alpha^2}8r^{-2}\bigl(\tfrac{z-q_0}{2z-2q_0+r^\alpha}\bigr)^2.
\end{align*}
The last expression combined with Condition~\ref{cond:short-range} 
leads us to the estimate \eqref{eq:1912171524}.
\end{proof}

We introduce a \emph{weight function} 
\begin{equation}\label{eq:weight-func}
\Theta = \bar\chi_{m+2}\theta^{2\beta};  \quad 
\theta = \int_0^{f/2^\nu}(1+s)^{-1-\delta}{\rm d}s=[1-(1+f/2^\nu)^{-\delta}]/\delta,
\end{equation}
where $\beta, \delta>0$ and $\nu\in\N_0$.
If we denote the derivatives of $\theta$ in $f$ by primes, we have 
\begin{equation*}
\theta'=(1+f/2^\nu)^{-1-\delta}/2^\nu,\quad 
\theta''=-(1+\delta)(1+f/2^\nu)^{-2-\delta}/2^{2\nu}.
\end{equation*}
We note that on $\mathop{\mathrm{supp}}\Theta$ it holds that 
\begin{equation}\label{eq:2001062031}
\partial_i\partial_j f 
= r^{-\alpha/2-1}\delta_{ij}-(1+\alpha/2)r^{\alpha/2-1}(\partial_if)(\partial_jf), 
\end{equation}
and we also note the function $\theta$ has the following properties.
\begin{lemma}\label{lem:theta-inequality}
Fix any $\delta>0$ in \eqref{eq:weight-func}.
Then there exist $c, C, C_k>0,\ k=2, 3, \ldots$, such that for any $k=2,3,\ldots$ 
and uniformly in $\nu\in\N_0$
\begin{align*}
&
c/2^\nu \le \theta \le \min\{C, f/2^\nu\},
\\&
c(\min\{2^\nu, f\})^\delta f^{-1-\delta}\theta \le \theta' \le f^{-1}\theta,
\\&
0\le(-1)^{k-1}\theta^{(k)} \le C_kf^{-k}\theta.
\end{align*}
\end{lemma}
We omit the proof.
We only refer \cite{IS1, Ita2} for the details.

The following lemma is a key to prove Theorem~\ref{thm:rc-complex}.
\begin{lemma}\label{lem:rc-bound-key-est}
Let $\beta\in(0,1+\alpha/(1-\alpha/2))$. 
Fix any $\delta\in(0,\alpha/(1-\alpha/2))$ in \eqref{eq:weight-func}.
Then there exist $c,C>0$ such that uniformly in $z\in I_\pm$ and $\nu\in\N_0$, 
as quadratic forms on $\vD(H)$,
\begin{equation*}
\begin{split}
&\mathop{\mathrm{Im}}\bigl((A\mp a)^*\Theta(H_\alpha-z)\bigr) 
\\&\ge 
c(A\mp a)^*\bar\chi_{m+2}\theta'\theta^{2\beta-1}(A\mp a) 
+ cp_jf^{-1}\Theta\ell_{jk}p_k 
\\&\phantom{{}={}}{}
- Cf^{-1-2\min\{\rho+1/(1-\alpha/2), 1+2\alpha/(1-\alpha/2)\}+\delta}\theta^{2\beta} 
+\mathop{\mathrm{Re}}\bigl(\gamma\theta^{2\beta}(H_\alpha-z)\bigr),
\end{split}
\end{equation*}
where $\gamma=\gamma_{z,\nu}$ is a certain function 
satisfying 
$|\gamma|\le Cr^{-2-\alpha}f^{-1-2\rho+\delta}$.
\end{lemma}

\begin{proof}
Fix $\beta\in(0,1+\alpha/(1-\alpha/2))$ 
and $\delta\in(0,\alpha/(1-\alpha/2))$ as in the assertion.
By the expression \eqref{eq:1912161558} we can write
\begin{equation}\label{eq:2001061716}
\begin{split}
&\quad\hspace{-3mm}
2\mathop{\mathrm{Im}}\bigl((A\mp a)^*\Theta(H_\alpha-z)\bigr) 
\\&= 
\mathop{\mathrm{Im}}\bigl((A\mp a)^*\Theta(A\pm a)r^\alpha(A\mp a) \bigr) 
+ \mathop{\mathrm{Im}}\bigl((A\mp a)^*\Theta p_jr^\alpha\ell_{jk}p_k \bigr) 
\\&\phantom{{}={}}{}
+ 2\mathop{\mathrm{Im}}\bigl((A\mp a)^*\Theta q_2\bigr).
\end{split}
\end{equation}
Let us estimate each term.
By the bounds \eqref{eq:1912171445}, 
we can estimate the first term of \eqref{eq:2001061716} as 
\begin{equation}\label{eq:2001061803}
\begin{split}
&\quad\hspace{-3mm}
\mathop{\mathrm{Im}}\bigl((A\mp a)^*\Theta(A\pm a)r^\alpha(A\mp a) \bigr) 
\\&= 
\mathop{\mathrm{Im}}\bigl((A\mp a)^*\Theta Ar^\alpha(A\mp a) \bigr) 
\pm (A\mp a)^*\Theta\bigl(\mathop{\mathrm{Im}}a \bigr) r^\alpha(A\mp a)
\\&\ge 
\beta(A\mp a)^*\bar\chi_{m+2}\theta'\theta^{2\beta-1}(A\mp a)  
- \tfrac\alpha2(A\mp a)^*\Theta r^{\alpha/2-1}(A\mp a)
\\&\phantom{{}={}}{}
+ \tfrac\alpha2(A\mp a)^*\Theta r^{\alpha/2-1}(A\mp a)
- C_1(A\mp a)^*\Theta r^{-\alpha/2-1}(A\mp a).
\\&= 
\beta(A\mp a)^*\bar\chi_{m+2}\theta'\theta^{2\beta-1}(A\mp a)  
- C_1(A\mp a)^*\Theta r^{-\alpha/2-1}(A\mp a).
\end{split}
\end{equation}
As for the second term of \eqref{eq:2001061716} 
we first decompose into four terms as follows.
\begin{equation}\label{eq:2001061825}
\begin{split}
\mathop{\mathrm{Im}}\bigl((A\mp a)^*\Theta p_jr^\alpha\ell_{jk}p_k \bigr) 
&= 
\mathop{\mathrm{Im}}\bigl((A\mp a)^*p_j\Theta r^\alpha\ell_{jk}p_k \bigr) 
\\&= 
\mathop{\mathrm{Im}}\bigl(p_jA\Theta r^\alpha\ell_{jk}p_k \bigr) 
+ \mathop{\mathrm{Im}}\bigl([A, p_j]\Theta r^\alpha\ell_{jk}p_k \bigr) 
\\&\phantom{{}={}}{}
\mp p_j\bigl(\mathop{\mathrm{Im}} a^*\bigr)\Theta r^\alpha\ell_{jk}p_k 
\mp \mathop{\mathrm{Im}}\bigl([a^*, p_j]\Theta r^\alpha\ell_{jk}p_k \bigr). 
\end{split}
\end{equation}
We further compute and estimate each term of \eqref{eq:2001061825}.
Noting the equality \eqref{eq:2001062031}, 
we can compute the first and second terms as 
\begin{equation}\label{eq:2001062033}
\begin{split}
&\quad\hspace{-3mm}
\mathop{\mathrm{Im}}\bigl(p_jA\Theta r^\alpha\ell_{jk}p_k \bigr) 
+ \mathop{\mathrm{Im}}\bigl([A,p_j]\Theta r^\alpha\ell_{jk}p_k \bigr) 
\\&= 
\tfrac1{2i}p_j\left\{A\Theta r^\alpha\ell_{jk}-\ell_{jk}r^\alpha\Theta A\right\}p_k 
+ \mathop{\mathrm{Re}}\bigl(p_i(\partial_i\partial_jf)\Theta r^\alpha\ell_{jk}p_k \bigr) 
\\&= 
- \beta p_j\bar\chi_{m+2}\theta'\theta^{2\beta-1}\ell_{jk}p_k 
- \tfrac12p_j\bar\chi_{m+2}'\theta^{2\beta}\ell_{jk}p_k 
+ p_jr^{\alpha/2-1}\Theta\ell_{jk}p_k.
\end{split}
\end{equation}
By using the bounds \eqref{eq:1912171445}, 
we estimate the third term of \eqref{eq:2001061825} as 
\begin{equation}\label{eq:2001062204}
\mp p_j\bigl(\mathop{\mathrm{Im}} a^*\bigr)\Theta r^\alpha\ell_{jk}p_k 
\ge 
\tfrac\alpha2p_jr^{\alpha/2-1}\Theta\ell_{jk}p_k 
- C_2p_jr^{-\alpha/2-1}\Theta\ell_{jk}p_k.
\end{equation}
To estimate the fourth term of \eqref{eq:2001061825}, 
we use Condition~\ref{cond:short-range} and the Cauchy-Schwarz inequality. 
For any small $\varepsilon_1>0$ we have 
\begin{equation}\label{eq:2001071331}
\begin{split}
&\quad\hspace{-3mm}
\mp \mathop{\mathrm{Im}}\bigl([a^*, p_j]\Theta r^\alpha\ell_{jk}p_k \bigr) 
\\&= 
\mp \mathop{\mathrm{Re}}\bigl((\partial_j a)^*\Theta r^\alpha\ell_{jk}p_k \bigr) 
\\&\ge 
- \varepsilon_1 p_jf^{-1-\delta}\Theta\ell_{jk}p_k 
- C_3\varepsilon_1^{-1}(\partial_j a)^*\Theta f^{1+\delta}r^{2\alpha}\ell_{jk}(\partial_k a)
\\&\ge 
- \varepsilon_1 p_jr^{\alpha/2-1}\Theta\ell_{jk}p_k 
- C_4\varepsilon_1^{-1}f^{-1-2\rho}r^{-2}\Theta.
\end{split}
\end{equation}
We substitute the bounds \eqref{eq:2001062033}, \eqref{eq:2001062204} 
and \eqref{eq:2001071331} into \eqref{eq:2001061825}, and then we obtain 
\begin{equation*}
\begin{split}
&\quad\hspace{-3mm}
\mathop{\mathrm{Im}}\bigl((A\mp a)^*\Theta p_jr^\alpha\ell_{jk}p_k \bigr) 
\\&\ge 
(1+\alpha/2-\varepsilon_1)p_jr^{\alpha/2-1}\Theta\ell_{jk}p_k 
- \beta p_j\bar\chi_{m+2}\theta'\theta^{2\beta-1}\ell_{jk}p_k 
\\&\phantom{{}={}}{}
- C_2p_jr^{-\alpha/2-1}\Theta\ell_{jk}p_k 
- \tfrac12p_j\bar\chi_{m+2}'\theta^{2\beta}\ell_{jk}p_k 
- C_4\varepsilon_1^{-1}f^{-1-2\rho}r^{-2}\Theta.
\end{split}
\end{equation*}
Let us fix $\varepsilon_1\in(0,(1+\alpha/2)-(1-\alpha/2)\beta)$.
Then we have 
\begin{equation}\label{eq:2001071550}
\mathop{\mathrm{Im}}\bigl((A\mp a)^*\Theta p_jr^\alpha\ell_{jk}p_k \bigr) 
\ge 
c_1p_jf^{-1}\Theta\ell_{jk}p_k 
- C_5Q,
\end{equation}
where 
\begin{equation*}
Q = r^{-2}f^{-1-2\min\{\rho,(3\alpha/2)/(1-\alpha/2)\}+\delta}\theta^{2\beta} 
+ p_j f^{-1-2\rho+\delta}r^{-2-\alpha}\theta^{2\beta}\delta_{jk}p_k.
\end{equation*}
Let $\varepsilon_2>0$. 
Using the Cauchy-Schwarz inequality and \eqref{eq:1912171524}, 
we can estimate the third term of \eqref{eq:2001061716} by 
\begin{equation}\label{eq:2001072321}
\begin{split}
2\mathop{\mathrm{Im}}\bigl((A\mp a)^*\Theta q_2\bigr) 
&\ge 
- \varepsilon_2(A\mp a)^*f^{-1-\delta}\Theta(A\mp a) 
- C_6\varepsilon_2^{-1}f^{1+\delta}|q_2|^2\Theta
\\&\ge 
- \varepsilon_2(A\mp a)^*f^{-1-\delta}\Theta(A\mp a) 
- C_7\varepsilon_2^{-1}Q.
\end{split}
\end{equation}
By combining \eqref{eq:2001061716}, \eqref{eq:2001061803}, 
\eqref{eq:2001071550} and \eqref{eq:2001072321} we have 
\begin{equation*}
\begin{split}
&\quad\hspace{-3mm}
2\mathop{\mathrm{Im}}\bigl((A\mp a)^*\Theta(H_\alpha-z)\bigr) 
\\&\ge 
\beta(A\mp a)^*\bar\chi_{m+2}\theta'\theta^{2\beta-1}(A\mp a)  
- \varepsilon_2(A\mp a)^*f^{-1-\delta}\Theta(A\mp a) 
\\&\phantom{{}={}}{}
- C_1(A\mp a)^*\Theta r^{-\alpha/2-1}(A\mp a) 
+ c_1p_jf^{-1}\Theta\ell_{jk}p_k 
- (C_5+C_7\varepsilon_2^{-1})Q.
\end{split}
\end{equation*}
By taking $\varepsilon_2>0$ small enough, we can obtain 
\begin{equation}\label{eq:2001080023}
\begin{split}
&\quad\hspace{-3mm}
2\mathop{\mathrm{Im}}\bigl((A\mp a)^*\Theta(H_\alpha-z)\bigr) 
\\&\ge 
c_2(A\mp a)^*\bar\chi_{m+2}\theta'\theta^{2\beta-1}(A\mp a)  
+ c_1p_jf^{-1}\Theta\ell_{jk}p_k 
- C_8Q.
\end{split}
\end{equation}
Finally we estimate the remainder term $Q$ as 
\begin{equation}\label{eq:2001091001}
Q 
\le 
C_9r^{-2}f^{-1-2\min\{\rho,(3\alpha/2)/(1-\alpha/2)\}+\delta}\theta^{2\beta} 
+ \mathop{\mathrm{Re}}\left( f^{-1-2\rho+\delta}r^{-2-\alpha}\theta^{2\beta}(H_\alpha-z)\right).
\end{equation}
By substituting \eqref{eq:2001091001} into \eqref{eq:2001080023} 
we can obtain the desired bounds.
\end{proof}

\begin{proof}[Proof of Theorem~\ref{thm:rc-complex}]
We consider only for the upper sign for simplicity.
If $\beta=0$, the bounds \eqref{eq:2001091149} follow immediately 
from Theorem~\ref{thm:lap-bound}.
We let $\beta\in(0,\beta_c)$ and 
take any $\delta\in(0,2\min\{\rho+1/(1-\alpha/2),1+2\alpha/(1-\alpha/2)\}-2\beta)\cap(0,\alpha/(1-\alpha/2)), \psi\in f^{-\beta}\vB$ and $z\in I_+$.
Then by Lemma~\ref{lem:rc-bound-key-est}, the Cauchy-Schwarz inequality, 
Theorem~\ref{thm:lap-bound} and Lemma~\ref{lem:theta-inequality}
\begin{equation}\label{eq:2001141532}
\begin{split}
&
\bigl\|\bar\chi_{m+2}^{1/2}\theta'^{1/2}\theta^{\beta-1/2}(A-a)R(z)\psi\bigr\|^2 
+ \bigl\langle p_jf^{-1}\Theta\ell_{jk}p_k \bigr\rangle_{R(z)\psi} 
\\&\le 
C_1\Bigl[ 
\bigl\|\Theta^{1/2}(A-a)R(z)\psi\bigr\|_{\vB^*}\|\theta^\beta\psi\|_\vB 
\\&\phantom{{}={}C_1\Bigl[ }{}
+ \bigl\|f^{-(1+2\min\{\rho+1/(1-\alpha/2),2\alpha/(1-\alpha/2)\}-\delta)/2}\theta^\beta R(z)\psi\bigr\|^2
\\&\phantom{{}={}C_1\Bigl[ }{}
+
\bigl\|f^{-(1+2\rho+1/(1-\alpha/2)-\delta)/2}\theta^\beta R(z)\psi\bigr\| 
\bigl\|f^{-(1+2\rho+1/(1-\alpha/2)-\delta)/2}\theta^\beta\psi\bigr\|
\Bigr]
\\&\le 
C_22^{-2\beta\nu}\Bigl[ 
\bigl\|\bar\chi_m^{1/2}f^\beta(A-a)R(z)\psi\bigr\|_{\vB^*}\|f^\beta\psi\|_\vB 
+ \|f^\beta\psi\|_\vB^2\Bigr].
\end{split}
\end{equation}
By commuting $R(z)$ and powers of $f$, we can see 
$f^\beta(A-a)R(z)\psi\in\vB^*$ for each $z\in I_+$.
Thus the right-hand side of \eqref{eq:2001141532} is finite, and then it follows that 
\begin{equation}\label{eq:2001151412}
\begin{split}
&
2^{2\beta\nu}\bigl\|\bar\chi_{m+2}^{1/2}\theta'^{1/2}\theta^{\beta-1/2}(A-a)R(z)\psi\bigr\|^2 
+ 2^{2\beta\nu}\bigl\langle p_jf^{-1}\Theta\ell_{jk}p_k \bigr\rangle_{R(z)\psi} 
\\&\le 
C_2\Bigl[ 
\bigl\|\bar\chi_m^{1/2}f^\beta(A-a)R(z)\psi\bigr\|_{\vB^*}\|f^\beta\psi\|_\vB 
+ \|f^\beta\psi\|_\vB^2\Bigr].
\end{split}
\end{equation}
In the first term on the left-hand side of \eqref{eq:2001151412}, 
we restrict the integral region to $\{2^\nu\le f<2^{\nu+1}\}$ 
and take supremum in $\nu\in\N$.
Then we obtain 
\begin{equation*}
\begin{split}
&\quad\hspace{-3mm}
c_1\bigl\|\bar\chi_{m+2}^{1/2}f^\beta(A-a)R(z)\psi\bigr\|^2 
\\&\le 
C_2\Bigl[ 
\bigl\|\bar\chi_m^{1/2}f^\beta(A-a)R(z)\psi\bigr\|_{\vB^*}\|f^\beta\psi\|_\vB 
+ \|f^\beta\psi\|_\vB^2\Bigr].
\end{split}
\end{equation*}
By the Cauchy-Schwarz inequality we can deduce 
\begin{equation}\label{eq:2001151455}
\bigl\|\bar\chi_{m+2}^{1/2}f^\beta(A-a)R(z)\psi\bigr\| 
\le 
C_3\|f^\beta\psi\|_\vB^2.
\end{equation}
As for the second term on the left-hand side of \eqref{eq:2001151412}, 
we apply the bound \eqref{eq:2001151455} and take limit  $\nu\to\infty$, 
and then we obtain by the Lebesgue's monotone convergence theorem 
and the concavity of $\theta$
\begin{equation}\label{eq:2001151510}
\bigl\langle p_j\bar\chi_{m+2}f^{2\beta-1}\ell_{jk}p_k \bigr\rangle_{R(z)\psi}^{1/2} 
\le 
C_4\|f^\beta\psi\|_\vB.
\end{equation}
Taking into account Theorem~\ref{thm:lap-bound}, 
we can replace the cut-off $\bar\chi_{m+2}$ of the bounds \eqref{eq:2001151455} 
and \eqref{eq:2001151510} to $1$.
Hence we are done.
\end{proof}

\section*{Acknowledgement}

The author would like to express his gratitude to 
Kenichi~Ito and Erik~Skibsted 
for constructive comments on the draft of the paper 
and warm encouragement.



\begin{thebibliography}{99}
\bibitem{ACH} S.~Agmon, J.~Cruz, I.~Herbst, 
Generalized Fourier transform for Sch\"odinger operators with potantials of order zero,\ 
J.\ Funct.\ Anal. 167 (1999) 345--369.

\bibitem{AH} S.~Agmon, L.~H\"ormander, 
Asymptotic properties of solutions of differential equations with simple characteristics,\ 
J.\ d'Anal.\ Math.\ 30 (1976) 1--38.

\bibitem{AIIS} T.~Adachi, K.~Itakura, K.~Ito, E.~Skibsted, 
Spectral theory for 1-body Stark operators,\ 
J.\ Differential Equations 268 (2020) 5179--5206.

\bibitem{BCHM} J.~F.~Bony, R.~Carles, D.~H\"afner, L.~Michel, 
Scattering theory for the Schr\"odinger equation with repulsive potential,\ 
J.\ Math.\ Pures Appl.\ 84 (2005) 509--579.

\bibitem{BM} V.~S.~Buslaev, V.~B.~Matveev, 
Wave operators for the Schr\"odinger equation with a slowly decreasing potential,\ 
Theor.\ Math.\ Phys.\ 2 (1970) 266--274, (English trans.\ from Russian). 

\bibitem{GY} Y.~G\^atel, D.~Yafaev, 
On solutions of the Schr\"odinger equation with radiation conditions at infinity: the long-range case,\ 
Ann.\ Inst.\ Fourier, Grenoble.\ 49, 5 (1999) 1581--1602.

\bibitem{Ho} L.~H{\"o}rmander, 
The existence of wave operators in scattering theory,\ 
Math.\ Z.\ 146 (1976) 69--91. 

\bibitem{HS} B.~Helffer, J.~Sj\"ostrand, 
Equation de Schr\"odinger avec champ magn\'etique et \'equation de Harper,\ 
Lecture notes in Phys.\ 345, Schr\"odinger operators, 118--197, eds. H.~Holden, A~Jensen, Springer, Berlin--Heidelberg--New York (1989).

\bibitem{Ike} T.~Ikebe, 
Spectral representations for Schr\"odinger operators with long--range potentials,\ 
J.\ Functional\ Analysis.\ 20 (1975) 158--177.

\bibitem{II} T.~Ikebe, H.~Isozaki, 
A stationary approach to the existence and completeness of long-range  wave operators,\ 
Integral Equations and Operator Theory\ 5 (1982) 18--49. 

\bibitem{Ishi} A.~Ishida, 
On inverse scattering problem for the Schr\"odinger equation with repulsive potentials,\ 
J.\ Math.\ Phys.\ 55 (2014) no. 8, 082101, 12 pp. 

\bibitem{Iso} H.~Isozaki, 
Eikonal equations and spectral representations for long--range Schr\"odinger Hamiltonians,\ 
J.\ Math.\ Kyoto Univ.\ 20 (1980) 243--261.

\bibitem{Ita1} K.~Itakura,
Rellich's theorem for spherically symmetric repulsive Hamiltonians,\ 
Math.\ Z.\ 291 (2019) no.\ 3, 1435--1449.

\bibitem{Ita2} K. Itakura,
Limiting absorption principle and radiation condition for repulsive Hamiltonians,\ 
To appear in Funkcial.\ Ekvac.

\bibitem{IS1} K.~Ito, E.~Skibsted, 
Radiation condition bounds on manifolds with ends,\ 
J.\ Funct.\ Anal.\ 278 (2020) no.\ 9, 108449

\bibitem{IS2} K.~Ito, E.~Skibsted, 
Stationary scattering theory on manifolds,\ 
To appear on Ann.\ Inst.\ Fourier\ (Grenoble) 

\bibitem{JP} A. Jensen, P. Perry, 
Commutator methods and Besov space estimates for {S}chr\"odinger operators,\ 
J.\ Operator Theory 14 (1985) 181--188.

\bibitem{M} R.~B.~Melrose, 
Spectral and scattering theory for the Laplacian on asymptotically Euclidian spaces,\ 
Marcel Dekker (1994) pp.\ 85--130.

\bibitem{MZ} R.~B.~Melrose, M.~Zworski, 
Scattering metrics and geodesic flow at infinity,\ 
Invent.\ Math.\ 124 (1996) 389--436. 

\bibitem{Mo} E.~Mourre, 
Absence of singular continuous spectrum for certain selfadjoint operators,\ 
Comm.\ Math.\ Phys.\ 78 (1980/81) no.~3, 391--408.

\bibitem{N} F.~Nicoleau, 
Inverse scattering for a Scr\"odinger operator with a repulsive potential,\ 
Acta Math.\ Sin.\ (Engl. Ser.) 22 (2006) no. 5, 1485--1492.

\bibitem{RS} M.~Reed, B.~Simon, 
Methods of modern mathematical physics. II. Fourier analysis, self-adjointness,\ 
Academic Press, New York-London, 1975. 

\bibitem{Sa} Y.~Saito, 
On the S-matrix for Schrodinger operators with long--range potentials,\ 
J.\ reine Angew.\ Math.\ 314 (1980) 99--116. 

\bibitem{Sk} E. Skibsted, 
Renormalized two-body low-energy scattering,\ 
Journal d'Analyse Math\'e-matique 122 (2014) 25--68.

\bibitem{Y} K.~Yosida, 
Functional analysis,\ 
Springer--Verlag, 1966.
\end{thebibliography}
\end{document}